%% file: main.tex
\documentclass[letterpaper, 10 pt, conference]{ieeeconf}  %

\IEEEoverridecommandlockouts                              %

\usepackage[T1]{fontenc}
\usepackage{titlesec}
\usepackage{amsmath, amssymb, graphicx, subfigure, enumerate}

\usepackage{amsthm}
\usepackage{alltt} 
\usepackage{tikz}
\usetikzlibrary{calc,decorations.pathmorphing,shapes.geometric}
\usepackage{graphicx,ctable,booktabs}
\usepackage{mathtools}
\usepackage[boxed, noend]{algorithm2e}
\makeatletter
\let\NAT@parse\undefined
\makeatother
\usepackage[hidelinks]{hyperref}
\usepackage{centernot}
\usepackage{bbm}
\usepackage{stmaryrd} %
\usepackage{yfonts}

\theoremstyle{definition}
\newtheorem{theorem}{Theorem}
\newtheorem{definition}[theorem]{Definition}

\newtheorem{proposition}[theorem]{Proposition}

\newtheorem{example}[theorem]{Example}

\theoremstyle{definition}

\newcommand{\defi}[1]{\emph{#1}}

\title{\LARGE \bf Efficient Reachable Sets on Lie Groups Using Lie Algebra Monotonicity and Tangent Intervals}

\author{Akash Harapanahalli and Samuel Coogan%
\thanks{*This work was supported in part by the Air Force Office of Scientific Research under grant FA9550-23-1-0303 and the National Science Foundation under award \#2219755.}%
\thanks{Akash Harapanahalli and Samuel Coogan are with the School of Electrical and Computer Engineering, Georgia Institute of Technology, Atlanta, GA, USA, 30318. \{\texttt{aharapan},\texttt{sam.coogan}\}\texttt{@gatech.edu}}%
}

\input{macros}

\definecolor{dblue}{rgb}{.098,.243,.424}
\definecolor{dcompb}{RGB}{157,35,0}  %

\begin{document}

\maketitle
\thispagestyle{empty}
\pagestyle{empty}

\begin{abstract}
In this paper, we efficiently compute overapproximating reachable sets for control systems evolving on Lie groups, building off results from monotone systems theory and geometric integration theory. 
We consider intervals in the tangent space, which describe real sets on the Lie group through the exponential map.
A local equivalence between the original system and a system evolving on the Lie algebra allows existing interval reachability techniques to apply in the tangent space. 
Using interval bounds of the Baker-Campbell-Hausdorff formula, these reachable set estimates are extended to arbitrary time horizons in an efficient Runge-Kutta-Munthe-Kaas integration algorithm. 
The algorithm is demonstrated through consensus on a torus and attitude control on $SO(3)$.
\end{abstract}

\section{INTRODUCTION}

One way to verify the safe behavior of a complex control system is to overapproximate its reachable set, the set of all possible states the system might reach under uncertainties in initial conditions and input.
When working with safety-critical systems, computing such reachable sets in a computationally efficient manner can help to verify the satisfaction of goal and safety specifications.
There has been a growing body of work in efficiently computing reachable sets for control systems under uncertainty~\cite{XC-SS:22}, including several tools such as CORA~\cite{CORA} for polytope-based reachability, JuliaReach~\cite{JuliaReach} using Taylor model flowpipes, and DynIBEX~\cite{DynIBEX} for robust Runge-Kutta schemes.
However, almost every tool existing in the literature deals with systems evolving on Euclidean state spaces.

Many real mechanical systems evolve on manifolds, which locally resemble vector spaces, but have different global properties. For example, 
many robotics applications use Lie groups to model rigid body motions~\cite{JMS:13}, the attitude of a spacecraft~\cite{PC:84} evolves on $SO(3)$, and the Hamiltonian control system~\cite{PEC-AVDS:87} evolves on a symplectic manifold. 
Manifold state spaces have been studied extensively in the field of Geometric Control~\cite{FB-AL:05}---and working directly with manifolds often allows one to better capture the underlying geometric properties of a state space.

Monotone systems theory~\cite{DA-EDS:03} and mixed monotone systems theory~\cite{SC-MA:15} allows the computation of a reachable set at the relatively small cost of simulating only two trajectories of the system.
There has been some work in characterizing coordinate-free structures similar to monotonicity on manifolds.
Since manifolds are not generally globally orderable, cone fields~\cite{JL:89} provide a local notion of ordering.
Differential positivity~\cite{FF-RS:16}, the infinitesimal version of monotonicity, uses these cone fields to study a Perron-Frobenius theory.
Differential positivity has been studied in Lie groups~\cite{CM-RS:18a}, and homogeneous spaces~\cite{CM-RS:18b}.
However, these works study purely infinitesimal properties, and do not examine set-valued properties like reachable sets. 

There is a wealth of literature in geometric numerical integration techniques in Lie groups and Homogeneous spaces~\cite{AI-HMK-SN-AZ:00,EH-etal:06} that capture underlying geometric structure, such as the Runge-Kutta-Munthe-Kaas algorithm~\cite{HMK:98,HMK:99} and the Magnus expansion~\cite{SB-etal:09}.
In this paper, we seek to develop a similar geometric method for efficient reachable set computation on Lie groups.

\begin{figure}
    \centering
    \begin{tikzpicture}[thick,scale=1, every node/.style={scale=1}]
        \draw[fill=black!50, draw=none]
          (0, 0) to[out=20, in=140] (3, -0.5) to [out=60, in=160]
          (8, 1) to[out=130, in=75] cycle;
        
        \shade[thin, left color=black!10, right color=black!50, draw=none]
          (0, 0) to[out=10, in=140] (5, -1) to [out=60, in=190] 
          (8, 1) to[out=130, in=75] cycle;

        \coordinate (g) at (1.5,1) ;

        \fill[fill=blue!30,opacity=.2] (0,0.25) -- (6,-0.25) -- (7,2.75) -- (1,3) -- cycle;
        \node[rotate=0] at (1.75,2.5) {$T_{\mathring{x}}\calX\simeq\frakx$};
        \draw[->] (g) -- ++(1.5,-0.125);
        \draw[->] (g) -- ++(0.4, 1.2);

        \coordinate (DX0) at ($0.125*(6, -0.5)$);
        \coordinate (DY0) at ($0.175*(1, 3)$);
        \coordinate (ulT0) at ($(g) - 0.5*(DX0) - 0.5*(DY0)$) ;
        \coordinate (olT0) at ($(ulT0) + (DX0) + (DY0)$) ;
        \fill[draw=blue,fill=blue!30,opacity=.8] (ulT0) -- ++(DX0) -- ++(DY0) -- ++($-1*(DX0)$) -- cycle;
        \fill (ulT0) circle (2pt) ;
        \fill (olT0) circle (2pt) ;

        \coordinate (TT) at (4,2);
        \coordinate (DXT) at ($0.2*(6, -0.5)$);
        \coordinate (DYT) at ($0.3*(1, 3)$);
        \coordinate (ulTT) at ($(TT) - 0.5*(DXT) - 0.5*(DYT)$) ;
        \coordinate (olTT) at ($(ulTT) + (DXT) + (DYT)$) ;
        \coordinate (gT) at (4.5,1.25) ;

        \shade[thin, left color=red!30, right color=red!60, draw=red]
          ($(gT)+(-0.75,-0.25)$) to[out=0, in=140] ($(gT)+(0.4,-0.6)$) to [out=70, in=-150] 
          ($(gT)+(0.8,0.05)$) to[out=140, in=10] ($(gT)+(-0.35,0.4)$)
          to[out=-140, in=70]  cycle;

        \draw[->,thick] (TT) -- (gT) node[midway, right] {$\exp$};
        
        \fill[draw=blue,fill=blue!30,opacity=.8] (ulTT) -- ++(DXT) -- ++(DYT) -- ++($-1*(DXT)$) -- cycle;
        \fill (ulTT) circle (2pt) node[above left] {$\ul\Theta(t)$};
        \fill (olTT) circle (2pt) node[below right] {$\ol\Theta(t)$};
        \draw[->,thick] (ulT0) to[out=0, in=-120] (ulTT);
        \draw[->,thick] (olT0) to[out=0, in=-120] (olTT);

        \fill (g) circle (2pt) node[above left] {$\mathring{x}$};
        \node at (7.5, 0.5) {$\calX$} ;

    \end{tikzpicture}
    \caption{A pictoral representation of Theorems~\ref{thm:smalltreach} and~\ref{thm:embeddedLiealg}. 
    An interval in the tangent space $T_{\mathring{x}}\calX$, pictured in blue, is evolved to time $t$ using monotone systems theory in the Lie algebra. The tangent interval at $t$ is exponentiated to the Lie group, pictured in red, overapproximating the reachable set at $t$. 
    }
    \label{fig:enter-label}
\end{figure}
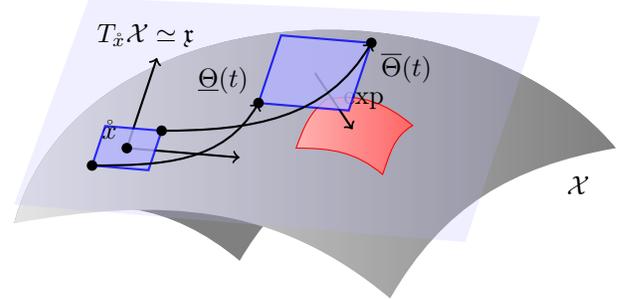

\emph{Contributions:} \ In this paper, we derive an approach inspired by monotone systems theory and geometric integration theory to efficiently overapproximate reachable sets for control systems evolving on Lie groups.
In Proposition~\ref{prop:canonicalcoords}, we discuss the local equivalence of a control system on a Lie group with a control system evolving on the Lie algebra.
Definition~\ref{def:tangentinterval} introduces the \emph{tangent interval}, which is an interval living in the tangent space to a manifold. Coupled with the exponential map, the tangent interval allows us to study real sets on the Lie group state space.
When the equivalent system is monotone, in Theorem~\ref{thm:smalltreach}, for small $t$, we can compute the reachable set from an exponentiated tangent interval by flowing two trajectories in the Lie algebra.
When the equivalent system is not monotone, but the cone in the Lie algebra is simplicial, in Theorem~\ref{thm:embeddedLiealg}, we can similarly compute the reachable set by flowing one trajectory of an embedding system of twice the dimension.
Using an inclusion function for the Baker-Campbell-Hausdorff (BCH) formula, Proposition~\ref{prop:recenterviabch} describes how we can ``recenter'' the tangent interval to extend this reachable set to any desired $T>0$, culminating in Algorithm~\ref{alg:RKMK-Reach}, which builds on the Runge-Kutta-Munthe-Kaas numerical integration method to compute overapproximating reachable sets of the system.
In Proposition~\ref{prop:abelian}, we show how this approach simplifies for abelian Lie groups.
Finally, in Section~\ref{sec:casestudies}, we demonstrate Algorithm~\ref{alg:RKMK-Reach} on two case studies: coupled oscillators evolving on a torus, and attitude control on $SO(3)$.

\section{NOTATION AND BACKGROUND} \label{sec:background}

\subsection{Interval analysis and (mixed) monotone systems}

A (pointed) \defi{cone} is a convex subset $K\subset V$ of a vector space $V$ such that $K + K \subset K$, $\lambda K \subset K$ for every $\lambda \geq 0$, and $K\cap -K = \{0\}$.
A pointed cone induces a partial order on $V$ as follows: for $x,y\in V$, $x\leq_K y \iff y - x\in K$. For $x\leq_K y$, denote the $K$-interval as $[x,y]_K = \{z\in V : x\leq_K z \leq_K y\}$. 
Let $\bbI_KV$ be the set of all $K$-intervals.

Let $V$ and $W$ be $n$- and $m$-dim vector spaces with cones $K$ and $C$ inducing partial orders $\leq_K$ and $\leq_C$. Given a map $f:V\to W$, the map $\sfF=[\ul\sfF,\ol\sfF]_C:\I_KV \to \I_CW$ is called an \defi{inclusion function} for $f$ if for every $v\in[\ulv,\olv]_K$,
\begin{align*}
    f(v) \in \sfF([\ulv,\olv]_K) = [\ul\sfF([\ulv,\olv]_K),\ol\sfF([\ulv,\olv]_K)]_C.
\end{align*}
$K$ and $C$ are called simplicial cones if they can be identified with the positive orthants $\R^n_+$ and $\R^m_+$ with the right bases, and in that case, $V$ and $W$ are equivalent to $\R^n$ and $\R^m$ with the standard elementwise ordering $\leq$.
We use $\IR^n$ to denote the set of $\R^n_+$-intervals, which are also represented without subscript as $[\cdot,\cdot]$.
This case, often called Interval Analysis, has been very well studied in the literature~\cite{LJ-MK-OD-EW:01} and allows for rapid robust computation of functions under input perturbations, such as floating point inaccuracies and sensor noise.
Inclusion functions can be composed together since their domain and codomain are both intervals---for example, $\sfF\circ\sfG$ is an inclusion function for $f\circ g$ if $\sfF$ and $\sfG$ are inclusion functions for $f$ and $g$.
This compositional approach, yielding what is called the \defi{natural inclusion function}, has been leveraged to build toolboxes for automated construction of inclusion functions, such as \verb|npinterval|~\cite{AH-SJ-SC:23a}.

Consider the nonlinear control system
\begin{align} \label{eq:monotonebanach}
    \dot{v} = f(v,u),
\end{align}
where $v\in V$ is the state, $u\in\calU\subset W$ is the control input in a subset $\calU$ of $W$, and $f$ is a parameterized vector field on $V$. The system~\eqref{eq:monotonebanach} is \defi{monotone} with respect to $(\leq_K,\leq_C)$ if
\begin{align*}
\begin{aligned}
    v_1 \leq_{K} v_2 &\text{ and } \bfu_1(t) \leq_{C} \bfu_2(t) \text{ a.e.} \\
    &\implies \phi_t(v_1, \bfu_1) \leq_K \phi_t(v_2, \bfu_2),
\end{aligned}
\end{align*}
where $\phi_t$ denotes the usual flow of~\eqref{eq:monotonebanach} to time $t$~\cite{DA-EDS:03}.
The defining property of a monotone system allows one to overapproximate its reachable set by simulating two extreme trajectories of the system.
In particular,
\begin{align*}
    v_0\in [\ulv_0,\olv_0]_{K} &\text{ and } \bfu(t) \in[\ul\bfu(t), \ol\bfu(t)]_C \text{ a.e.}\\
    &\implies v(t) \in [\ulv(t), \olv(t)]_{K},
\end{align*}
where $v(t) = \phi_t(v_0,\bfu)$, $\ulv(t) = \phi_t(\ulv_0,\ul\bfu)$, $\olv(t) = \phi_t(\olv_0,\ol\bfu)$. 

Now, consider the case where the dynamics~\eqref{eq:monotonebanach} are not monotone. Assume $K$ and $C$ are simplicial cones, so in the right bases, $V = \R^n$, $W = \R^m$, and $f:\R^n\times \R^m\to \R^n$.
If one can find an inclusion function\footnote{For the embedding system to be monotone, the inclusion function $\sfF$ should be monotone~\cite[p. 29]{LJ-MK-OD-EW:01}. The reachable set estimate holds regardless.} $\sfF=[\ul\sfF,\ol\sfF]:\IR^n\times\IR^m\to\IR^n$ for $f$, one can embed the dynamics~\eqref{eq:monotonebanach} into a \defi{mixed monotone embedding system} in $2n$ dimensions,
\begin{align} \label{eq:mmembsys}
\begin{aligned}
    \dot{\ulv}_i &= \ul\sfE([\ulv,\olv],[\ulu,\olu])_i := \ul\sfF([\ulv,\olv_{i:\ulv}],[\ulu,\olu])_i, \\
    \dot{\olv}_i &= \ol\sfE([\ulv,\olv],[\ulu,\olu])_i := \ol\sfF([\ulv_{i:\olv},\ulv],[\ulu,\olu])_i,
\end{aligned}
\end{align}
for each $i=1,\dots,n$, where $\smallconc{\ulv}{\olv}\in\R^{2n}$, $\ulv\leq\olv$, $\smallconc{\ulu}{\olu}\in\R^{2m}$, $\ulu\leq\olu$. 
The notation $v_{i:w}$ is the vector where $(v_{i:w})_j = v_j$ for every $j\neq i$ and $(v_{i:w})_i = w_i$---thus, $[\ulv,\olv_{i:\ulv}]$ ($[\ulv_{i:\olv},\olv]$ resp.) is the lower (upper resp.) $i$-th face of the axis-aligned hyperrectangle $[\ulv,\olv]$.
The system~\eqref{eq:mmembsys} is a monotone system with respect to the \defi{southeast orders} in $\R^{2n}$ and $\R^{2m}$, defined as $\smallconc{\ulv}{\olv} \leqse \smallconc{\ulw}{\olw}$ if and only if $\ulv \leq \ulw$ and $\olw \leq \olv$.
The trajectory of the embedding system similarly overapproximates the reachable set of the original system,
\begin{align*}
    v_0\in [\ulv_0,\olv_0] &\text{ and } \bfu(t) \in[\ul\bfu(t), \ol\bfu(t)] \text{ a.e.}\\
    &\implies v(t) \in [\ulv(t), \olv(t)],
\end{align*}
where $v(t) = \phi_t(v_0,\bfu)$ and $\smallconc{\ulv(t)}{\olv(t)} = \Phi_t\left(\smallconc{\ulv_0}{\olv_0},\smallconc{\ul\bfu}{\ol\bfu}\right)$, with $\Phi_t$ denoting the flow of~\eqref{eq:mmembsys}~\cite{SC-MA:15},~\cite[Proposition 5]{SJ-AH-SC:23}.

\subsection{Lie groups}

Let $M$ be a (smooth) manifold, $T_p M$ denote the tangent space at $p\in M$ and $TM:= \bigsqcup_{p\in M} T_pM$ denote the tangent bundle.
Given a smooth map $g:M\to N$ between manifolds, denote its differential map as $dg_p:T_pM\to T_{g(p)} N$. 
A vector field $f$ on $M$ is a possibly parameterized smooth section of the tangent bundle $TM$.

A \defi{Lie group} $G$ is a smooth manifold with group structure compatible with the manifold, \emph{i.e.}, the group operation is a smooth mapping from the product manifold $G\times G$ to $G$, and the group inverse is a smooth mapping from $G$ to $G$.
For simplicity, we will assume that $G$ is a \defi{matrix Lie group}, a subgroup of $GL(n)$ for some $n$.
For $g\in G$, define the \defi{left translation map} $L_g:G\to G$ such that $L_g(h) = gh$ for every $h\in G$. 
A Lie group $G$ is immediately associated with its \defi{Lie algebra} $\frakg$, the set of left-invariant vector fields on $G$. A vector field $f$ on $G$ is \defi{left-invariant} if for every $g,h\in G$,
\[f(gh) = (dL_g)_h(f(h)).\]
Every left-invariant vector field $f$ can be identified with a tangent vector $\Theta\in T_eG$ at the identity $e\in G$, as
\[f(g) = (dL_g)_e(\Theta),\]
so, equivalently, $\frakg= T_eG$.
The Lie algebra $\frakg$ is endowed with the \defi{Lie bracket} $\dbrak{\cdot,\cdot}$ from the typical Lie bracket of two vector fields, where we use double square brackets to avoid confusion with intervals.
The \defi{adjoint action} of a Lie algebra element $\Theta_1\in \frakg$ is the mapping $\operatorname{ad}_{\Theta_1}(\Theta_2) = \dbrak{\Theta_1,\Theta_2}$ for every $\Theta_2\in\frakg$.

An element $\Theta\in\frakg$ of the Lie algebra defines a unique one-parameter subgroup $\gamma_{\Theta}:\R\to G$, as the integral curve of the left-invariant vector field associated to $\Theta$ passing through $\gamma_\Theta(0) = e$. Define the \defi{exponential map} $\exp:\frakg\to G$, such that $\exp(\Theta) = \gamma_\Theta(1)$. 
Then, $\exp$ maps a neighborhood of $0\in \frakg$ diffeomorphically to a neighborhood of $e\in G$---denote this neighborhood $N_{\exp} \subset\frakg$. 
In the case of a matrix Lie group, the exponential map coincides with the usual matrix exponential,
\begin{align*}
    \exp(\Theta) = \sum_{k=0}^\infty \frac{\Theta^k}{k!}.
\end{align*}

The group structure allows us to identify every tangent space $T_gG \simeq \frakg$ with the left-translation map through \emph{left-trivialization}. For a vector field $f$ on $G$,
\begin{align*}
    f(g) &= d(L_g \circ L_{g^{-1}})_g (f(g)) \\
    &= (dL_g)_{e} (\underbrace{(dL_{g^{-1}})_g (f(g))}_{=:A(g)}) = gA(g),
\end{align*}
where $A : G\to\frakg$ and the last equality holds for matrix Lie groups. Any vector field is fully defined by such a left-trivialized mapping to the Lie algebra.

A \emph{cone field} on a manifold $M$ is a mapping $\calK$ where for every $p\in M$, $\calK(p)\subset T_pM$ is a cone in the tangent space at $p$.
A cone field $\calK$ on Lie group $G$ is left-invariant if 
\[\calK(gh) = (dL_{g})_{h} (\calK(h)),\]
for every $g,h\in G$.
Similar to how any left-invariant vector field can be identified with a tangent vector at the identity $e$,
every left-invariant cone field can be identified with a cone $K\subset \frakg$ at the tangent space at the identity as
\[
    \calK(g) = (dL_g)_e(K).
\]

\section{CONTROL SYSTEMS ON LIE GROUPS} \label{sec:controlonliegroups}

A control system evolving on a Lie group state space $\calX$ is a nonlinear vector field on $\calX$ parameterized by an input in a vector space. Since the vector field on can be left-trivialized by a mapping to the Lie algebra, without loss of generality, we assume that the control system is a tuple $\Sigma=(\calX,\calU,A)$, of a $n$-dimensional Lie group $\calX$, a subset $\calU\subset W$ of a $m$-dimensional vector space $W$, and a left-trivialized mapping $A:\calX\times\calU\to\mathfrak{x}=T_e\calX$ defining the dynamics,
\begin{align} \label{eq:liecontrolsys}
    \Sigma : \{\ \dot{x} = (dL_x)_e(A(x,u)) = xA(x,u), \ \ x(0) = x_0,
\end{align}
where $t\in[0,\infty)$, $x\in\calX$, and $u\in\calU$. 
Note that $\dot{x} = xA(x,u)$ holds for matrix Lie groups.
Equation~\eqref{eq:liecontrolsys} models a nonlinear control system on a Lie group.
For simplicity, we will assume \emph{forward completeness}, \emph{i.e.}, for a fixed measurable $\bfu:[0,\infty)\to\calU$, the solution to~\eqref{eq:liecontrolsys} is uniquely defined for $t\in[0,\infty)$, and let $\phi_t(x_0, \bfu)$ denote this trajectory from initial condition $x_0\in\calX$.

The Lie group structure allows us to define a locally equivalent system evolving on the Lie algebra $\frakx=T_e\calX$, using the differential of the exponential map.
We follow a similar treatment as~\cite{AI-HMK-SN-AZ:00}, with slight deviations for left-trivialization instead of right-trivialization.

\begin{definition}[Differential of $\exp$] \label{def:dexp} %
    For a Lie group $G$ with Lie algebra $\frakg$ and exponential map $\exp:\frakg\to G$, define the function $\dexp:\frakg\times\frakg\to\frakg$ as the left-trivialized differential of the exponential map, \emph{i.e.},
    \begin{align*}
        d(\exp)_\Theta(\Omega) &= \frac{d}{dt}\Big|_{t=0} \exp(\Theta + t\Omega) \\
        &= (dL_{\exp(\Theta)})_e(\dexp_{\Theta}(\Omega)) \\
        &= \exp(\Theta)\dexp_\Theta(\Omega),
    \end{align*}
    where the final equality holds for matrix Lie groups.
\end{definition}

For $\Theta\in\frakg$, $\dexp_\Theta$ has an analytic expression involving the adjoint action of $\Theta$,
\begin{align} \label{eq:dexp}
\begin{aligned}
    \dexp_{\Theta} &= \frac{1 - \exp(-\operatorname{ad}_\Theta)}{\operatorname{ad}_\Theta} = \sum_{k=0}^\infty \frac{(-1)^k}{(k+1)!}(\operatorname{ad}_\Theta)^k,
\end{aligned}
\end{align}
and its inverse, $\dexp^{-1}_\Theta$, can be obtained by inverting the analytic expression and taking its summation expansion,
\begin{align} \label{eq:dexpinv}
    \dexp^{-1}_{\Theta} = \frac{\operatorname{ad}_\Theta}{1 - \exp(-\operatorname{ad}_\Theta)} = \sum_{k=0}^\infty \frac{B_k}{k!} (\operatorname{ad}_\Theta)^k,
\end{align}
where $B_k$ are the Bernoulli numbers with $B_1 = \frac12$.
$\dexp^{-1}$ allows one to define a vector field on the Lie algebra which exactly characterizes the local behavior of the true dynamics evolving on the original Lie group.
In the literature~\cite{AI-HMK-SN-AZ:00}, this is called canonical coordinates of the first kind---numerical integration is done in a canonical basis for the tangent space arising from the differential of the exponential map, left translated to the centering point $\mathring{x} = x_0\exp(\Theta_0)^{-1}$.

\begin{proposition}[Canonical coordinates] \label{prop:canonicalcoords} %
    Consider the control system~\eqref{eq:liecontrolsys}.
    For small $t\geq 0$, the trajectory $t\mapsto x(t)$ from initial condition $x_0$ under measurable $\bfu:[0,\infty)\to\calU$ is given by 
    \begin{align} \label{eq:liealgsys}
    \begin{gathered}
        x(t) = x_0\exp(\Theta_0)^{-1}\exp(\Theta(t)), \\
        \Upsilon :\left\{
        \begin{aligned}
            \dot{\Theta}(t) &= \dexp_{\Theta(t)}^{-1}(A(x(t),\bfu(t))), \\
            \Theta(0) &= \Theta_0,
        \end{aligned}
        \right.
    \end{gathered}
    \end{align}
    for every $\Theta_0\in N_{\exp}$, $\Theta(t)\in\frakg$ and $\dexp^{-1}$ defined as~\eqref{eq:dexpinv}. 
\end{proposition}

The proof for Proposition~\ref{prop:canonicalcoords} is a small variation on~\cite[Lemma 3.1]{AI-HMK-SN-AZ:00} for Lie groups, which follows by time differentiation and uniqueness of $x(t)$; alternatively~\cite[Corollary 1]{HMK:98} applies in the more general setting of homogeneous spaces.

Proposition~\ref{prop:canonicalcoords} provides a natural approach to adapt tools for numerical integration on vector spaces to Lie groups. 
If one were to apply traditional integration techniques to the nonlinear space $\calX$, thought as a submanifold/subgroup of a known vector space, \emph{e.g.}, $\bbR^{n\times n}$, the numerical approximation error will quickly accumulate causing the solution to drift away from the manifold.
Instead, one can apply standard integration techniques to the system~\eqref{eq:liealgsys} evolving on the linear space $\frakg$, with a small enough step size, and the exponential map guarantees that the solution will remain on the manifold.
These types of geometric numerical integrators are explored in depth in~\cite{AI-HMK-SN-AZ:00}, and we utilize the Runge-Kutta-Munthe-Kaas integration scheme~\cite{HMK:98} in Algorithm~\ref{alg:RKMK-Reach} below.

\section{REACHABILITY VIA LIE ALGEBRA} \label{sec:monotoneliealgebra}

Proposition~\ref{prop:canonicalcoords} established a local equivalence between the control system~\eqref{eq:liecontrolsys} and the system~\eqref{eq:liealgsys}, allowing us to improve simulation capabilities using the linear Lie algebra.
In this section, we apply (mixed) monotone reachable set computations in the equivalent system~\eqref{eq:liealgsys} to obtain overapproximations for the original Lie group control system~\eqref{eq:liecontrolsys}.

\subsection{The tangent interval}

We propose to study the following object living in the tangent space to the manifold.

\begin{definition}[Tangent interval] \label{def:tangentinterval}
    Given a cone field $\calK$ on a smooth manifold $M$, a point $p\in M$, and two vectors $\ulv_p,\olv_p\in T_pM$ such that $\ulv_p \leq_{\calK(p)} \olv_p$, let
    \begin{align*}
        [\ulv_p,\olv_p]_{\calK(p)} := \{v_p\in T_pM : \ulv_p \leq_{\calK(p)} v_p \leq_{\calK(p)} \olv_p\},
    \end{align*}
    denote a \emph{tangent interval}.
\end{definition}

On its own, the tangent interval is not directly useful as it has no relation to the original manifold. However, when coupled with additional structure on the manifold, like the Lie group structure, this can become a useful object to study.

Suppose $G$ is a Lie group with a left-invariant cone field $\calK$ identified by the cone $K\subset\frakg$. One can left-trivialize a tangent interval $[\ulv_g,\olv_g]_{\calK(g)}$ by identifying the tangent space $T_gG$ with $\frakg$ in the usual way,
\begin{align*}
    [\ulv_g,\olv_g]_{\calK(g)} &= d(L_g \circ L_{g^{-1}})_g ([\ulv_g,\olv_g]_{\calK(g)}), \\
    &= (dL_g)_e (\underbrace{(dL_{g^{-1}})_g ([\ulv_g,\olv_g]_{\calK(g)})}_{=:[\ul\Theta,\ol\Theta]_K}) = g[\ul\Theta,\ol\Theta]_K,
\end{align*}
with $\ul\Theta,\ol\Theta\in\frakg$, $\ul\Theta \leq_K \ol\Theta$,
which follows since the cone field $\calK(g)$ is left-invariant. 
The last equality, read as the set image $\{g\Theta : \Theta\in[\ul\Theta,\ol\Theta]_K\}$, holds for matrix Lie groups.
With the exponential map, a tangent interval represents a real set on the Lie group, namely
\begin{align*}
    g\exp([\ul\Theta,\ol\Theta]_K) = \{g\exp(\Theta) : \ul\Theta \leq_K \Theta \leq_K \ol\Theta\},
\end{align*}
which we call the \emph{exponentiated tangent interval}.

\subsection{Monotone Lie algebra dynamics}

Equipped with the tangent interval, traditional monotone systems theory can be applied to the equivalent system~\eqref{eq:liealgsys}, evolving in the vector space Lie algebra.

\begin{theorem}[Monotone Lie algebra] \label{thm:smalltreach}
    Consider the control system~\eqref{eq:liecontrolsys}, and the associated system~\eqref{eq:liealgsys} in the Lie algebra. Let $\calK$ be a left-invariant cone field on $\calX$ identified by $K\subset\frakx$, and let $C$ be a cone on $W$. 
    Let $\mathring{x}\in\calX$.
    If $\Upsilon$ from~\eqref{eq:liealgsys} is monotone with respect to $(\leq_{K},\leq_{C})$, then for small $t\geq 0$,
    \begin{align*}
        x_0\in\mathring{x} \exp([\ul\Theta_0,\ol\Theta_0]_{K}) &\text{ and } \bfu(t)\in[\ul\bfu(t),\ol\bfu(t)]_C \text{ a.e.} \\
        & \implies x(t)\in \mathring{x}\exp([\ul\Theta(t),\ol\Theta(t)]_{K}),
    \end{align*}
    where $x(t) = \phi^\Sigma_t(x_0,\bfu)$, $\ul\Theta(t) = \phi^\Upsilon_t(\ul\Theta_0,\ul\bfu)$, $\ol\Theta(t) = \phi^\Upsilon_t(\ol\Theta_0,\ol\bfu)$.
\end{theorem}
\begin{proof}
    Let $x_0\in\mathring{x}\exp([\ul\Theta_0,\ol\Theta_0]_{K})$, which implies the existence of $\Theta_0\in[\ul\Theta_0,\ol\Theta_0]_K$ such that $x_0 = \mathring{x}\exp(\Theta_0)$.
    In particular, this implies that $\ul\Theta_0\leq_K\Theta_0$ and $\Theta_0\leq_K\ol\Theta_0$.
    By hypothesis,~\eqref{eq:liealgsys} is monotone, thus,
    \begin{align} \label{proof:Thetabound}
        \ul\Theta(t) \leq_K \Theta(t) \leq_K \ol\Theta(t),
    \end{align}
    where $\Theta(t) = \phi_t^\Upsilon(\Theta_0,\bfu)$. Using Proposition~\ref{prop:canonicalcoords}, for small $t\geq 0$,
    \begin{align*}
        x(t) =  x_0\exp(\Theta_0)^{-1} \exp(\Theta(t)) = \mathring{x} \exp(\Theta(t)),
    \end{align*}
    since $x_0 = \mathring{x}\exp(\Theta_0)$, completing the proof.
\end{proof}

Theorem~\ref{thm:smalltreach} allows us to simulate two trajectories in the Lie algebra to bound the reachable set on the original manifold for small $t>0$, provided $\Upsilon$ from~\eqref{eq:liealgsys} is monotone.

\subsection{Nonmonotone Lie algebra dynamics}
In the case that the control system~\eqref{eq:liealgsys} is not monotone, when $K$ and $C$ are simplicial cones, one can instead build a mixed monotone embedding system.
Recall that in the case of simplicial cones, the cone $K$ can be identified with $\R^n_+$, and the vector space $\frakg$ can be identified with $\R^n$. 
Let $\hat{\cdot}:\R^n\to\frakg$ denote this identification, and let $\cdot^\vee:\frakg\to\R^n$ denote the inverse map. 
Under these identifications, \emph{i.e.}, $\Theta = \hat{v}$, the system~\eqref{eq:liealgsys} can be written as follows,
\begin{align} \label{eq:euclididentify}
\begin{aligned}
    \dot{v}(t) &= \dot{\Theta}(t)^\vee \\ &= \underbrace{\dexp_{\hat{v}}^{-1}(A(x_0\exp(\hat{v}_0)^{-1}\exp(\hat{v}),u))^\vee}_{=:f(v,u)},
\end{aligned}
\end{align}
where $v\in\R^n$, $u\in\R^m$, and $f:\R^n\times\R^m\to\R^n$.\footnote{For notational brevity, we simply set $W=\R^m$ and $C=\R^m_+$ to avoid tracking explicit basis mappings for the control input.}
With an inclusion function $\sfF:\IR^n\times\IR^m\to\IR^n$ for $f$, the embedding system~\eqref{eq:mmembsys} can be used to obtain interval bounds on the reachable set of the Lie algebra system~\eqref{eq:liealgsys}.

\begin{theorem}[Embedded Lie algebra] \label{thm:embeddedLiealg}
Consider the control system~\eqref{eq:liecontrolsys}, and the associated system~\eqref{eq:liealgsys} in the Lie algebra. Let $\calK$ be a left-invariant cone field on $\calX$ identified by simplicial cone $K\subset\frakx$, and let $C$ be a simplicial cone on $W$.
Let $\mathring{x}\in\calX$.
If $\sfE$~\eqref{eq:mmembsys} is the embedding system induced by an inclusion function $\sfF$ on~\eqref{eq:euclididentify}, then for small $t\geq 0$,
    \begin{align*}
        x_0\in\mathring{x} \exp([\ul\Theta_0,\ol\Theta_0]_{K}) &\text{ and } \bfu(t)\in[\ul\bfu(t),\ol\bfu(t)] \text{ a.e.} \\
        & \implies x(t)\in \mathring{x}\exp([\ul\Theta(t),\ol\Theta(t)]_{K}),
    \end{align*}
    where $x(t) = \phi^\Sigma_t(x_0,\bfu)$, $t\mapsto\smallconc{\ulv(t)}{\olv(t)}$ is the trajectory $\Phi_t\Big(\smallconc{\ul\Theta_0^\vee}{\ol\Theta_0^\vee},\smallconc{\ul\bfu}{\ol\bfu}\Big)$ of the embedding system~\eqref{eq:mmembsys}, and $\ul\Theta(t) = \hat{\ulv}(t)$, $\ol\Theta(t) = \hat{\olv}(t)$.
\end{theorem}
\begin{proof}
    The proof is identical to that of Theorem~\ref{thm:smalltreach}, with a slight modification to obtain the bound~\eqref{proof:Thetabound}.
    Since \eqref{eq:mmembsys} is an embedding system for~\eqref{eq:euclididentify}, 
    \begin{align*}
        \ulv(t) \leq v(t) \leq \olv(t),
    \end{align*}
    by~\cite[Proposition 5]{SJ-AH-SC:23}, which implies that 
    \begin{align*}
        \ul\Theta(t) \leq_K \Theta(t) \leq_K \ol\Theta(t),
    \end{align*}
    where $\Theta(t) = \phi_t^\Upsilon(\Theta_0, \bfu)$ and $v(t) = \Theta(t)^\vee$. 
\end{proof}

In practice, obtaining the inclusion function $\sfF$ is possibly automatable using interval analysis toolboxes like \verb|npinterval|~\cite{AH-SJ-SC:23a}; however, these may involve matrix exponentials, infinite summations, and nonlinearities.

Theorems~\ref{thm:smalltreach} and~\ref{thm:embeddedLiealg} are only valid as long as the interval $[\ul\Theta,\ol\Theta]_K$ remains inside the neighborhood $N_{\exp}$, which cannot be guaranteed for any desired time $T>0$.

\subsection{Recentering via Baker-Campbell-Hausdorff (BCH)}

To extend Theorems~\ref{thm:smalltreach} and~\ref{thm:embeddedLiealg} to arbitrary time horizons, we use an interval overapproximation of the BCH formula to ``recenter'' the tangent interval to ensure it remains in $N_{\exp}$.
Let $G$ be a Lie group, with Lie algebra $\frakg = T_eG$ and exponential map $\exp:\frakg\to G$. Given $\Theta_1,\Theta_2\in\frakg$, the BCH formula solves for $\Theta_3$ satisfying
\begin{align} \label{eq:bchproblem}
    \exp(\Theta_1)\exp(\Theta_2) = \exp(\Theta_3).
\end{align}
The non-commutativity of $G$ implies that the solution will generally not be $\Theta_3 = \Theta_1 + \Theta_2$.  Indeed, the BCH formula famously provides the solution as
an infinite sum of commutators in the Lie algebra,
\begin{align}\label{eq:bch}
    \Theta_3 = &\operatorname{bch}_{\Theta_1}(\Theta_2) = \Theta_1 + \Theta_2 + \frac12\dbrak{\Theta_1,\Theta_2} \\
    &+ \frac1{12} \dbrak{\Theta_1,\dbrak{\Theta_1,\Theta_2}} - \frac1{12} \dbrak{\Theta_2,\dbrak{\Theta_1,\Theta_2}} + \cdots. \nonumber
\end{align}
For some Lie group/Lie algebra pairs, the BCH formula may admit a simpler closed-form solution, as is the case for, \emph{e.g.}, the Lie group $SO(3)$~\cite[Appendix B]{AI-HMK-SN-AZ:00}. 

Now, given a cone $K\subset\frakg$ inducing the partial order $\leq_K$, we assume availability of an inclusion function $\textsf{BCH}_{\Theta_1}:\I_K\frakg\to\I_K\frakg$ for any $\Theta_1\in\frakg$, such that for every $\Theta_2\in[\ul\Theta_2,\ol\Theta_2]_K\in\I_K\frakg$,
\begin{align}\label{eq:BCHincl}
    \operatorname{bch}_{\Theta_1}(\Theta_2) \in \textsf{BCH}_{\Theta_1}([\ul\Theta_2,\ol\Theta_2]_K).
\end{align}
In practice, this may be obtained in closed-form, or approximated using a sufficient truncation of~\eqref{eq:bch} and off the shelf interval analysis techniques.

\begin{proposition}[Recentering via $\textsf{BCH}$] \label{prop:recenterviabch}
    Let $G$ be a Lie group, with Lie algebra $\frakg = T_eG$. For every $\mathring{g}\in G$, $[\ul\Theta,\ol\Theta]_K\in\I_K\frakg$, and $\Theta\in\frakg$,
    \begin{align*}
        \mathring{g} \exp([\ul\Theta,\ol\Theta]_K) \subset \mathring{g}\exp(\Theta)\exp(\textsf{BCH}_{-\Theta}([\ul\Theta,\ol\Theta]_K)).
    \end{align*}
\end{proposition}
\begin{proof}
    Let $\Theta'\in[\ul\Theta,\ol\Theta]_K$. Since $\exp(\Theta)^{-1} = \exp(-\Theta)$,
    \begin{align*}
        \mathring{g}\exp(\Theta') 
        &= \mathring{g} \exp(\Theta) \exp(-\Theta) \exp(\Theta') \\
        &= \mathring{g} \exp(\Theta) \exp(\operatorname{bch}_{-\Theta}(\Theta')) \\
        &\in \mathring{g} \exp(\Theta) \exp(\textsf{BCH}_{-\Theta}([\ul\Theta,\ol\Theta]_K)),
    \end{align*}
    since $\textsf{BCH}_{-\Theta}$ is an inclusion function for $\operatorname{bch}_{-\Theta}$.
\end{proof}

\begin{figure}
    \centering
    \includegraphics[width=\columnwidth]{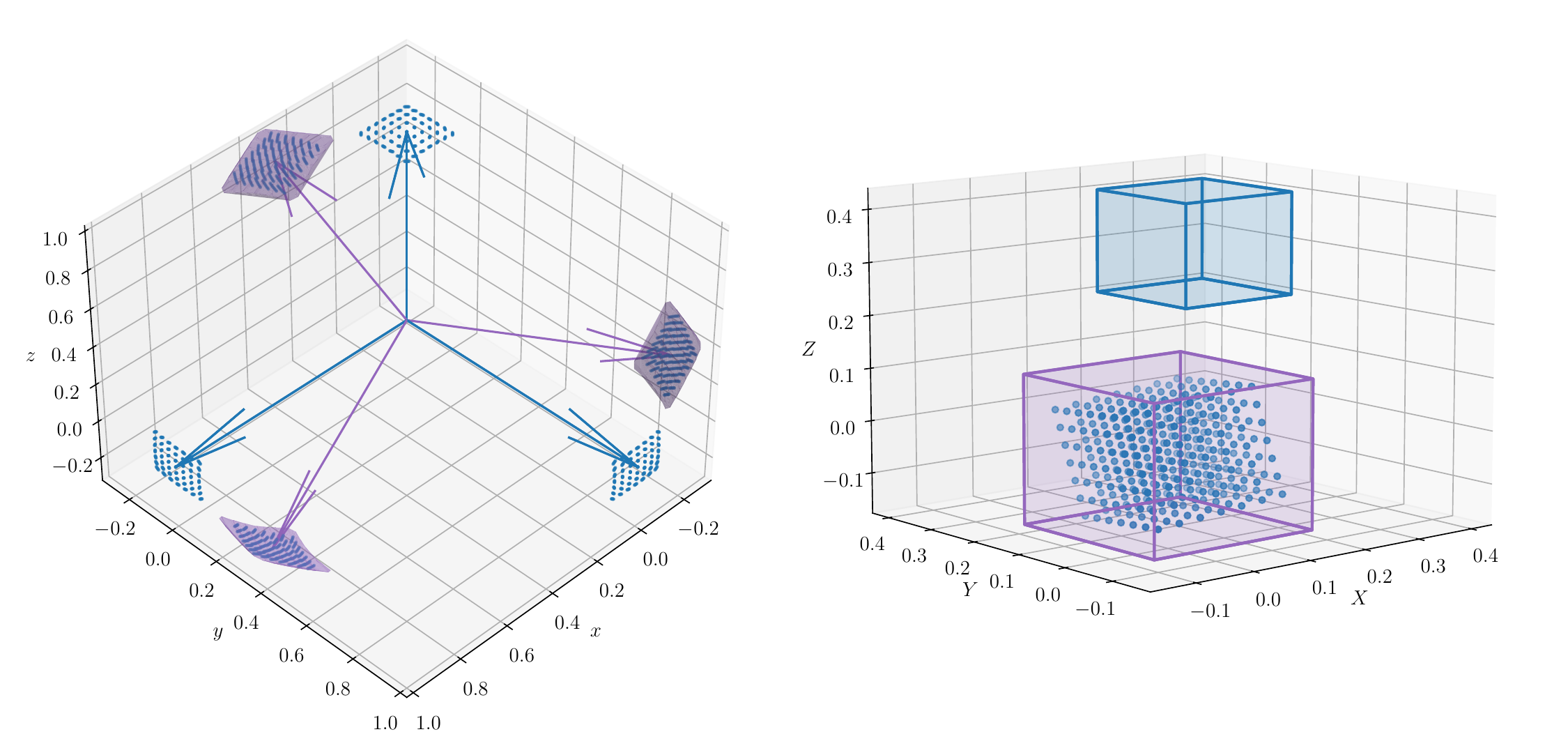}
    \caption{A pictoral representation of recentering using a $\textsf{BCH}$ inclusion function~\eqref{eq:BCHincl} for the Lie group $SO(3)$, with Lie algebra $\so(3)$ and the simplicial cone $K$ from Section~\ref{subsec:exSO3}. 
    \textbf{Left ($SO(3)$):}
    The exponentiated tangent intervals $I\exp([-0.1,0.1]^3)$ and $I\exp([0.2,0.4]^3)$ centered around the identity matrix are pictured in blue, using an evenly spaced meshgrid of $7^3=343$ points exponentiated from the Lie algebra.
    With $\mathring{\Theta} = [0.3\ 0.3\ 0.3]^T$,
    the exponentiated tangent interval $\exp(\mathring{\Theta}) \exp(\textsf{BCH}_{-\mathring{\Theta}}([0.2,0.4]^3))$ centered around $\exp(\mathring{\Theta})$ is pictured in purple.
    Visually, the geometry of the tangent interval centered at the identity before the BCH overapproximation is different than the geometry of the tangent interval centered at $\exp(\mathring{\Theta})$.
    \textbf{Right ($\so(3)$):} The box $[0.2,0.4]^3$ and the map $\Theta \mapsto \operatorname{bch}_{-\mathring{\Theta}}(\Theta)$ sampled with $7^3=343$ evenly spaced points in this box are pictured in blue. The box $\textsf{BCH}_{-\mathring{\Theta}}([0.2,0.4]^3)\approx[-1.4,1.4]^3$ from the BCH inclusion function is pictured in purple, overapproximating these outputted points.
    The exponential map links the left and right plots.
    }
    \label{fig:SO3compare}
\end{figure}

In Figure~\ref{fig:SO3compare}, we demonstrate this recentering procedure in $SO(3)$.
The inclusion function $\textsf{BCH}$ can be used to patch together valid reachable sets obtained using Theorems~\ref{thm:smalltreach} and~\ref{thm:embeddedLiealg}, to extend the reachable set to any arbitrary time $T > 0$. 
For example, a simple strategy would be to iteratively build reachable sets in the following manner. As input, take a left-trivialized exponentiated tangent interval $\mathring{x}_t\exp([\ul\Theta_t,\ol\Theta_t]_K)$.
\begin{enumerate}[(i)]
    \item Simulate the lower and upper bounds of the monotone system or the embedding system to $t + \Delta t$, obtaining $\ul\Theta_{t + \Delta t}$ and $\ol\Theta_{t + \Delta t}$. The reachable set at $t + \Delta t$ is $\mathring{x}_t \exp([\ul\Theta_{t+\Delta t},\ol\Theta_{t + \Delta t}]_K)$.
    \item Select the midpoint $\mathring{\Theta} = \frac{\ul\Theta_{t + \Delta t} + \ol\Theta_{t + \Delta t}}{2}$, and its associated Lie group element $\mathring{x}_{t + \Delta t} = \mathring{x}_t\exp(\mathring{\Theta} )$ for the next canonical centering.
    \item Use Proposition~\ref{prop:recenterviabch} to obtain an overapproximation $\mathring{x}_{t + \Delta t}\exp(\textsf{BCH}_{-\mathring{\Theta}}([\ul\Theta_{t+\Delta t},\ol\Theta_{t + \Delta t}]_K))$.
\end{enumerate}
In Algorithm~\ref{alg:RKMK-Reach}, we implement this iterative scheme numerically, using the Runge-Kutta method in the Lie algebra~\cite[Ch 3]{AI-HMK-SN-AZ:00}, \cite{HMK:98}. 
We refer to~\cite{JCB:64} for an in-depth discussion about Butcher tableaus.
As long as each $[\ul\Theta_n,\ol\Theta_n]_K\subset N_{\exp}$, Theorem~\ref{thm:smalltreach}, Theorem~\ref{thm:embeddedLiealg}, and Proposition~\ref{prop:recenterviabch} guarantee that Algorithm~\ref{alg:RKMK-Reach} returns a valid overapproximated reachable set on the Lie group, up to small numerical inaccuracies.

\RestyleAlgo{ruled}
\SetKwInput{KwInput}{Input}
\SetKwInput{KwParam}{Param}
\SetKwInput{KwAssum}{Assumptions}
\begin{algorithm}[t]
\caption{Runge-Kutta-Munthe-Kaas Reachability in the Lie Algebra}
\label{alg:RKMK-Reach}
\KwInput{Initial set $\mathring{x}_0\exp([\ul\Theta_0,\ol\Theta_0]_{K})$; control bounds $\ul\bfu,\ol\bfu:[0,\infty)\to\calU$, left-trivialized dynamics $A:\calX\times\calU\to\frakx$; time step $h>0$; horizon $N$}
\KwParam{Butcher tableau $\{a_{kl},b_l,c_k\}_{1\leq k,l \leq \nu}$}
\SetInd{0em}{0.8em}
\For{$n=0,\dots,N-1$}{
    $\ul\Omega_0 = \ul\Theta_{n}$; $\ol\Omega_0 = \ol\Theta_{n}$\;
    \For{$k=1,\dots,\nu$}{
        $\ul\Omega_k = \ul\Omega_0 + \sum_{l=1}^\nu a_{kl} \ulF_l$;
        $\ol\Omega_k = \ol\Omega_0 + \sum_{l=1}^\nu a_{kl} \olF_l$\;
        \If{$\Upsilon$~\eqref{eq:liealgsys} is monotone}{
            Use Theorem~\ref{thm:smalltreach}:\\
            $\ulF_k = \dexp^{-1}_{\ul\Omega_k}(hA(\mathring{x}_n\exp(\ul\Omega_k),\ul\bfu(t_n+c_kh)))$\;
            $\olF_k = \dexp^{-1}_{\ol\Omega_k}(hA(\mathring{x}_n\exp(\ol\Omega_k),\ol\bfu(t_n+c_kh)))$\;
        } \ElseIf {$K$, $C$ simplicial} {
            Use Theorem~\ref{thm:embeddedLiealg}:\\
            \vspace{0.25em}
            $\smallconc{\ulF_k^\vee}{\olF_k^\vee} = h\sfE\left(\smallconc{\ul\Omega_k^\vee}{\ol\Omega_k^\vee}, \smallconc{\ul\bfu(t_n+c_kh)}{\ol\bfu(t_n+c_kh)}\right)$\;
        }
    }
    $\ul\Omega = \sum_{l=1}^\nu b_{l} \ulF_l$;
    $\ol\Omega = \sum_{l=1}^\nu b_{l} \olF_l$\;

    \uIf{\textsf{BCH} recentering condition} {
        $\Omega =  \frac{\ul\Omega + \ol\Omega}{2}$;
        $\mathring{x}_{n+1}\gets \mathring{x}_{n}\exp\left(\Omega\right)$\;
        $[\ul\Theta_{n+1}, \ol\Theta_{n+1}]_K \gets \textsf{BCH}_{-\Omega}\left([\ul\Omega, \ol\Omega]_K\right)$\;
    } \Else {
        $\mathring{x}_{n+1}\gets \mathring{x}_{n}$\;
        $[\ul\Theta_{n+1}, \ol\Theta_{n+1}]_K \gets [\ul\Omega, \ol\Omega]_K$\;
    }
}
\Return Reachable set $\left\{\mathring{x}_n\exp([\ul\Theta_{n}, \ol\Theta_{n}]_{K})\right\}_{n=0}^N$;
\end{algorithm}

\subsection{Abelian Lie groups}

In an abelian (commutative) Lie group $G$, the Lie bracket is identically zero. A direct implication is, for $\Theta_1,\Theta_2\in\frakg$,
\begin{align} \label{eq:abelianBCH}
    \exp(\Theta_1)\exp(\Theta_2) = \exp(\Theta_1 + \Theta_2),
\end{align}
which follows by, \emph{e.g.}, a straightforward computation from the BCH formula~\eqref{eq:bch}. 
The following Proposition characterizes some interesting behavior for abelian Lie groups.

\begin{proposition}[Abelian Lie groups] \label{prop:abelian}
    Let $G$ be an abelian Lie group, with a left-invariant cone field $\calK$ identified by $K\subset\frakg$. The following statements hold:
    \begin{enumerate}[(i)]
        \item The left-trivialized $\dexp_\Theta$ and $\dexp^{-1}_\Theta$ is the identity map for every $\Theta\in\frakg$; \label{prop:abelian:p1}
        \item The map $\textsf{BCH}_{\Theta_1}([\ul\Theta_2,\ol\Theta_2]_K) = [\Theta_1 + \ul\Theta_2, \Theta_1 + \ol\Theta_2]_K$ is an inclusion function for $\operatorname{bch}_{\Theta}$. \label{prop:abelian:p2}
    \end{enumerate}
\end{proposition}
\begin{proof}
    Regarding (i), let $\Theta,\Omega\in\frakg$. Since $G$ is Abelian,
    \begin{align*}
        \frac{d}{dt}\Big|_{t=0} \exp(\Theta + t\Omega) 
        &= \frac{d}{dt} \Big|_{t=0} \exp(\Theta)\exp(t\Omega) \\
        &= \exp(\Theta)\frac{d}{dt} \Big|_{t=0} \exp(t\Omega) 
    \end{align*}
    Since $\exp(t\Omega)$ the same as the integral curve $\gamma:\R\to G$ of the left-invariant vector field identified by $\Omega$ passing through $\gamma(0) = e$, $\frac{d}{dt} \big|_{t=0} \exp(t\Omega) = \gamma'(0) = \Omega$. Thus,
    \begin{align*}
        \frac{d}{dt}\Big|_{t=0} \exp(\Theta + t\Omega) &= \exp(\Theta)\Omega.
    \end{align*}
    Thus, the left-trivialized $\dexp$ is $\dexp_\Theta(\Omega) = \Omega$, which in turn implies that $\dexp^{-1}_\Theta$ is also the identity map.

    Regarding (ii), let $\Theta_1\in\frakg$, and $\Theta_2\in[\ul\Theta_2,\ol\Theta_2]_K\in\I_K\frakg$. Since the Lie bracket is identically zero, using~\eqref{eq:bch}
    \begin{align*}
        \operatorname{bch}_{\Theta_1}(\Theta_2) = \Theta_1 + \Theta_2 \in [\Theta_1 + \ul\Theta_2, \Theta_1 + \ol\Theta_2]_K,
    \end{align*}
    completing the proof.
\end{proof}

In the case of an abelian state space, Proposition~\ref{prop:abelian} simplifies Algorithm~\ref{alg:RKMK-Reach}. Statement~\eqref{prop:abelian:p1} removes the need to evaluate the $\dexp^{-1}$ function to define the dynamics~\eqref{eq:liealgsys}, which also makes it simpler to verify the monotonicity of $\Upsilon$.
Statement~\eqref{prop:abelian:p2} provides a simple characterization of the \textsf{BCH} inclusion function as a translational shift in the Lie algebra, which simplifies the recentering process.
Additionally, the bidirectionality of~\eqref{prop:abelian:p2} guarantees no additional overapproximation error in the recentering step, since
\begin{align*}
    \textsf{BCH}_{-\Theta_1} ([\Theta_1 + \ul\Theta_2, \Theta_1 + \ol\Theta_2]_K) = [\ul\Theta_2,\ol\Theta_2]_K.
\end{align*}

\section{CASE STUDIES} \label{sec:casestudies}
\stepcounter{footnote}\footnotetext{Runtimes reported on computer with a Ryzen 5 5600X and 32 GB RAM. \\ \url{https://github.com/gtfactslab/Harapanahalli_CDC2024}}
\subsection{Monotonicity on the $2$-torus}
Consider the torus identified with the abelian matrix group $SO(2)^2$ and the consensus dynamics
\begin{align*}
    \dot{x}_1 = x_1\left(\hat{\omega}_1 + \log(x_2x_1^{-1})\right), \ \dot{x}_2 = x_2\left(\hat{\omega}_2 + \log(x_1x_2^{-1})\right),
\end{align*}
where $x_1,x_2\in SO(2)$, $\log:SO(2)\to\so(2)$, $\omega_1,\omega_2\in \R$, and the hat identification $\hat{\cdot}:\R\to\so(2)$ such that $\hat{\omega} = \begin{bsmallmatrix}
    0 & -\omega \\
    \omega & 0
\end{bsmallmatrix}$. Define the vee map $\cdot^\vee:\so(2)\to\R$ as the inverse of the hat map.
The exponential map $\exp:\so(2)\to SO(2)$ is injective on the neighborhood $N_{\exp} = \widehat{(-\pi,\pi)^2}$. 

The equivalent Lie algebra system~\eqref{eq:liealgsys} around center $\mathring{x}$ is
\begin{align} \label{eq:torusliealg}
\begin{aligned}
    \dot{\theta}^\vee_1 &= \omega_1 + \log(\mathring{x}_2 \exp(\theta_2) \mathring{x}_1^{-1} \exp(-\theta_1))^\vee \\
    \dot{\theta}^\vee_2 &= \omega_2 + \log(\mathring{x}_1 \exp(\theta_1) \mathring{x}_2^{-1} \exp(-\theta_2))^\vee
\end{aligned}
\end{align}
where $\theta_1,\theta_2\in N_{\exp}$. It can be shown that~\eqref{eq:torusliealg} is monotone as long as $(\log(\mathring{x}_2) + \theta_2 - \log(\mathring{x}_1) -\theta_1)^\vee \in (-\pi, \pi)$, which occurs when 
$\log(\mathring{x}_2), \theta_2, \log(\mathring{x}_1), \theta_1$
all lie within a single tangent interval of less than $\pi$ width.

Consider the left-invariant cone field $\calK$ induced by $K = (\widehat{\R_+})^2$.
We apply Algorithm~\ref{alg:RKMK-Reach}, with $\omega_1 = 5$, $\omega_2 = 2$, $x_1(0) \in \exp\left(\widehat{\frac{\pi}{2}}\right)\exp([-\widehat{0.6},\widehat{0.6}]_K)$ and $x_2(0) \in \exp\left(\widehat{\pi}\right)\exp([-\widehat{0.1},\widehat{0.1}]_K)$. 
We check to ensure the system is indeed monotone at every time step. 
Finally, since $SO(2)^2$ is abelian, we recenter at every time step at no loss.
Figure~\ref{fig:torus} shows the reachable set computed to $T = 3s$, which took $0.039 \pm 0.002$ seconds to compute, averaged over $100$ runs.

\begin{figure}
    \centering
    \includegraphics[width=0.24\columnwidth]{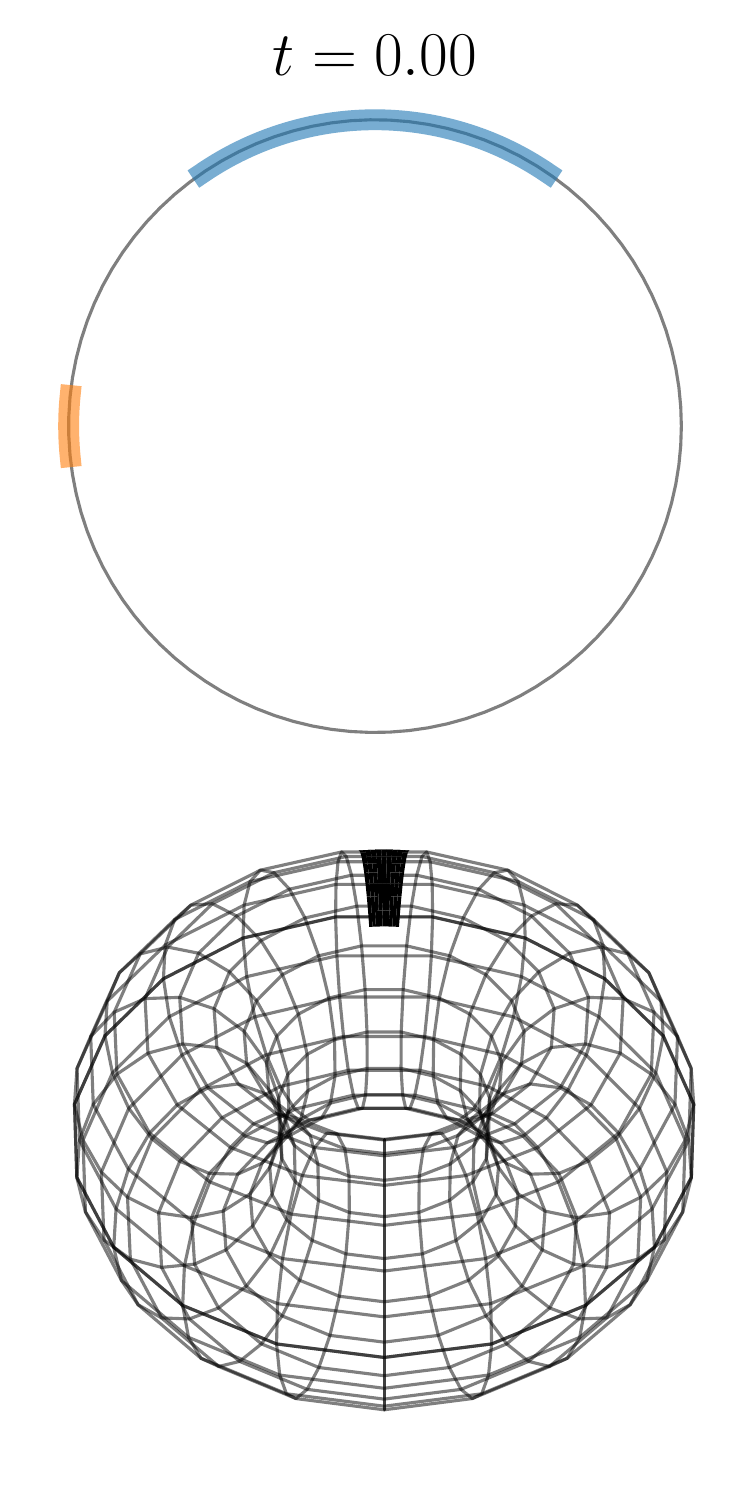}
    \includegraphics[width=0.24\columnwidth]{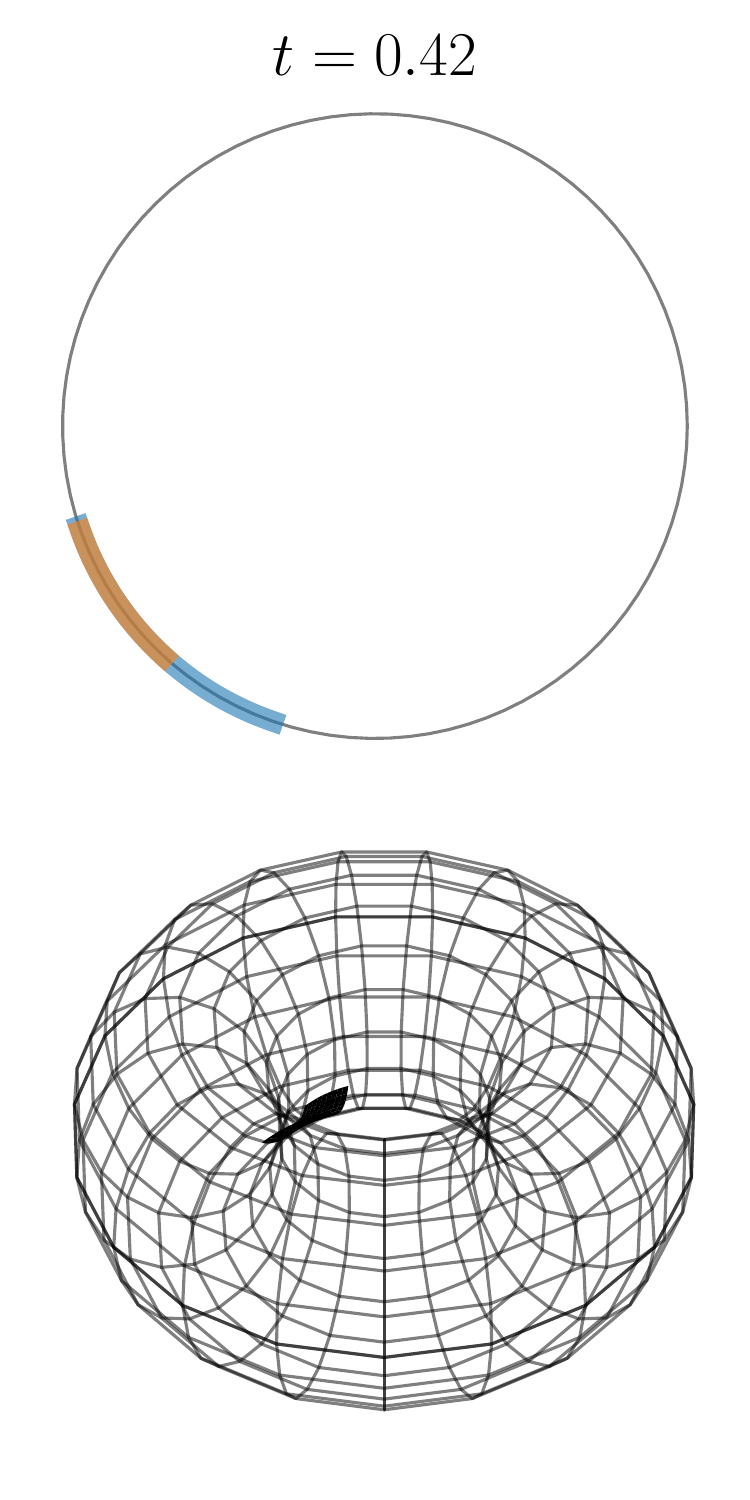}
    \includegraphics[width=0.24\columnwidth]{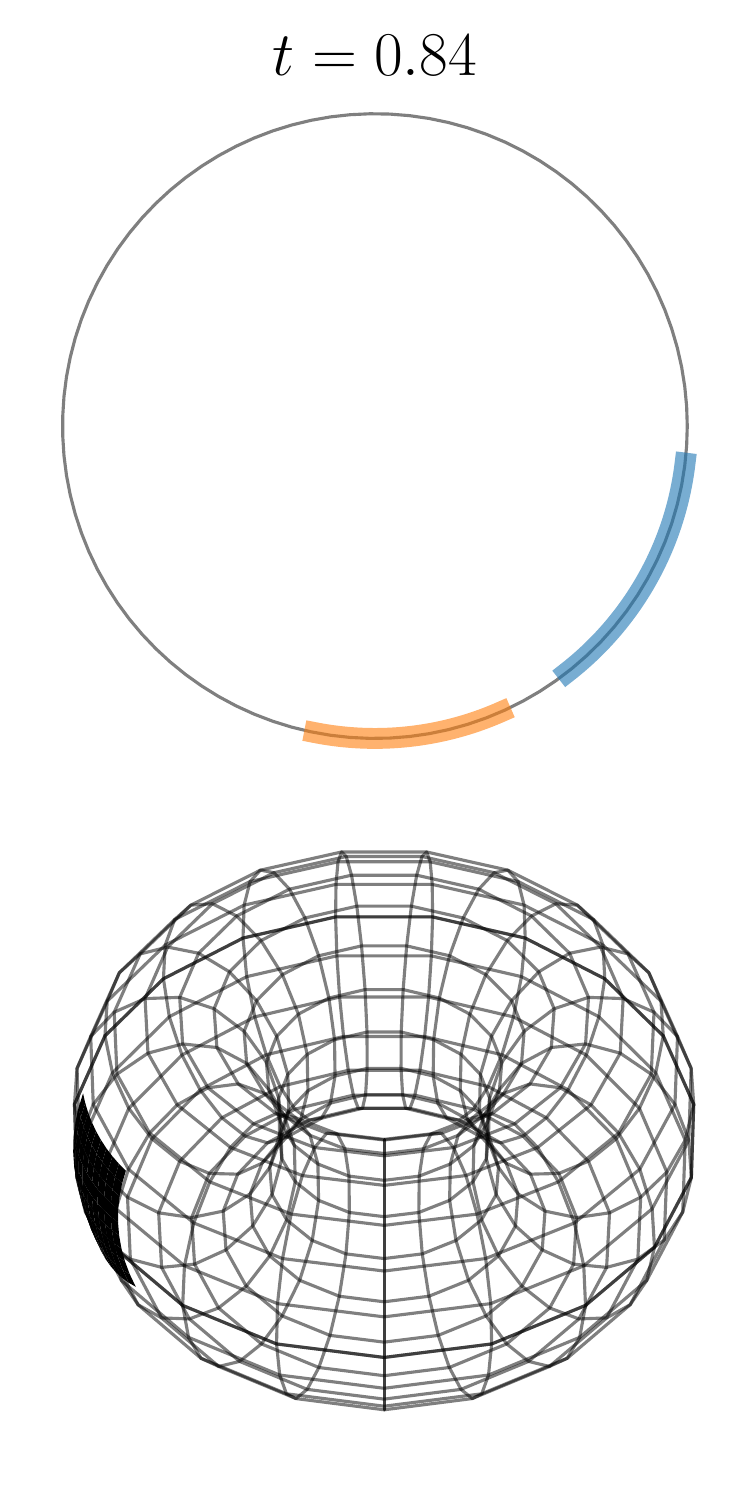}
    \includegraphics[width=0.24\columnwidth]{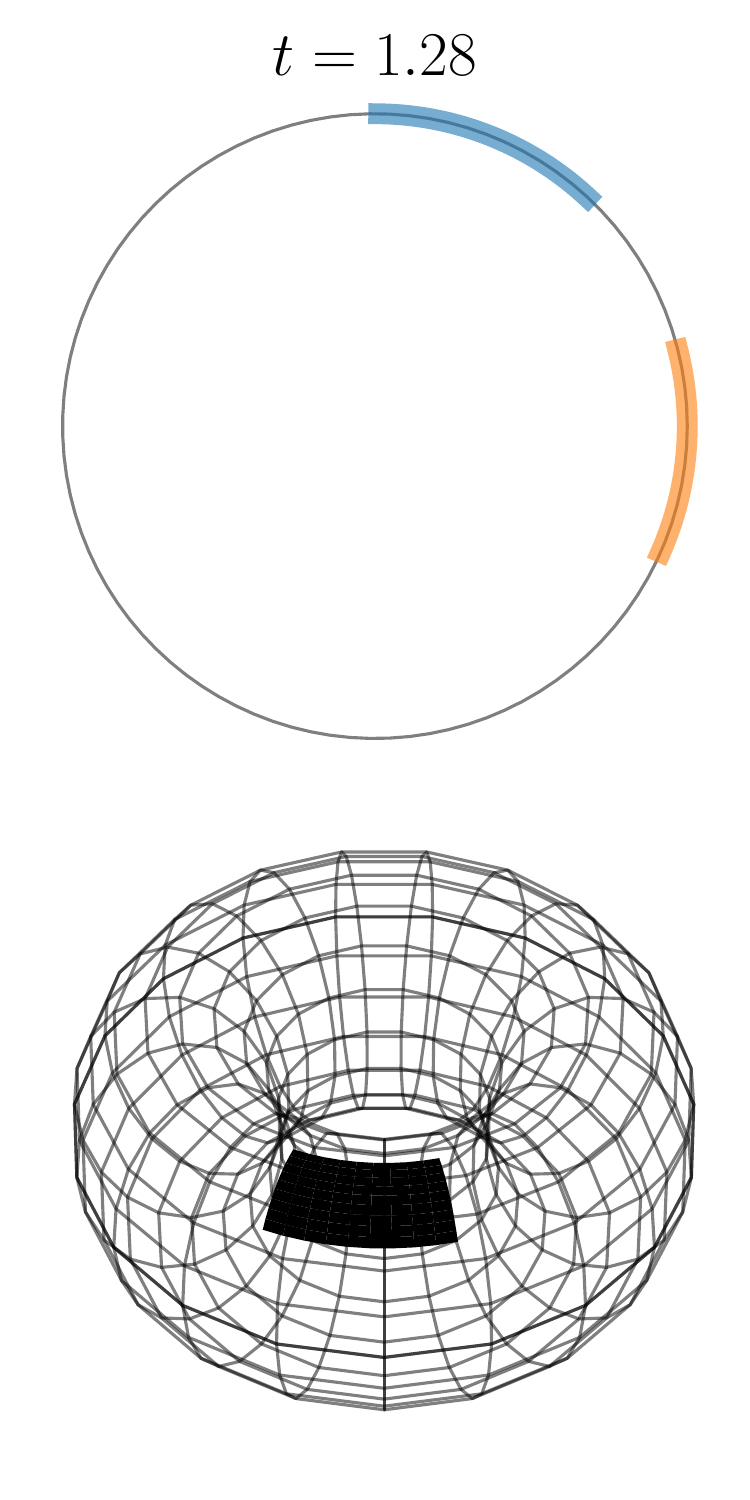}
    \includegraphics[width=0.24\columnwidth]{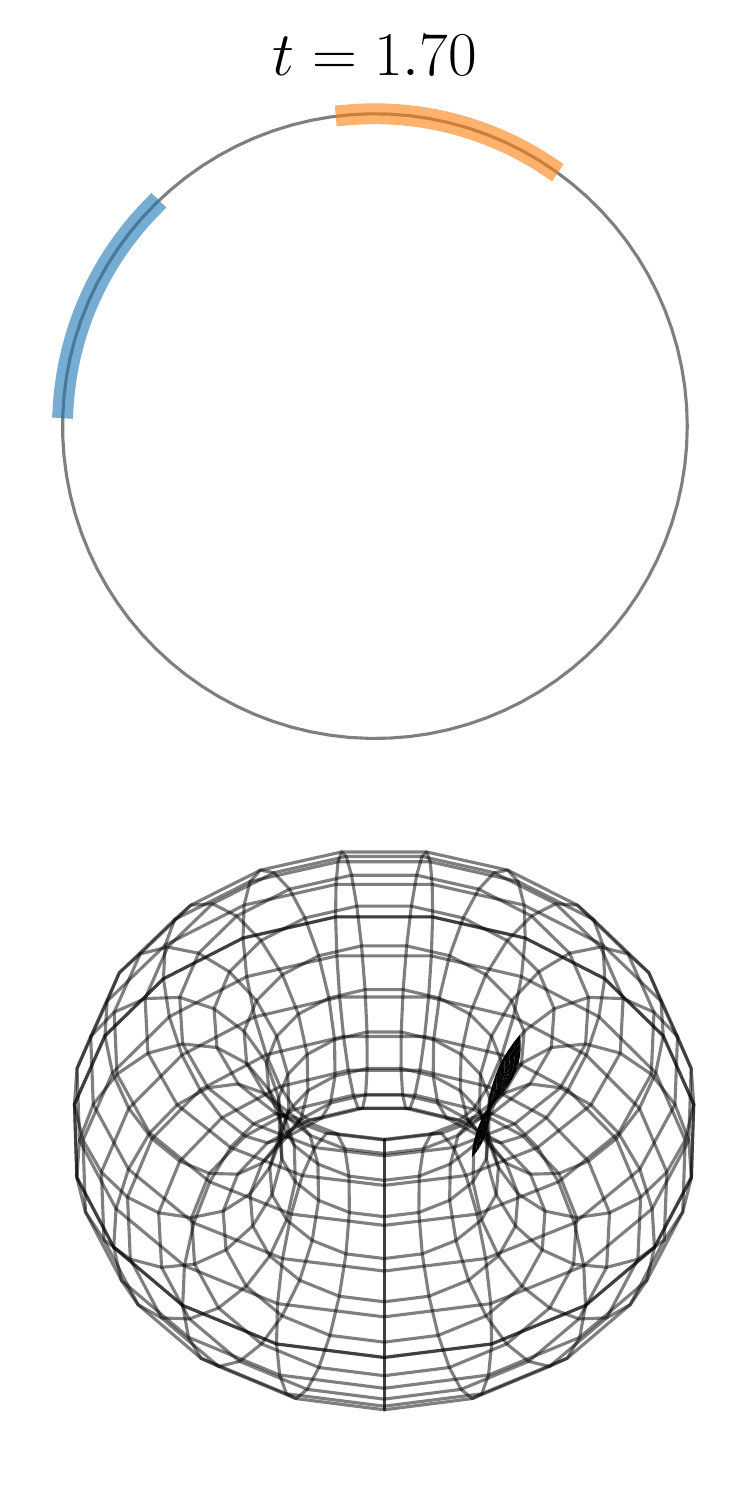}
    \includegraphics[width=0.24\columnwidth]{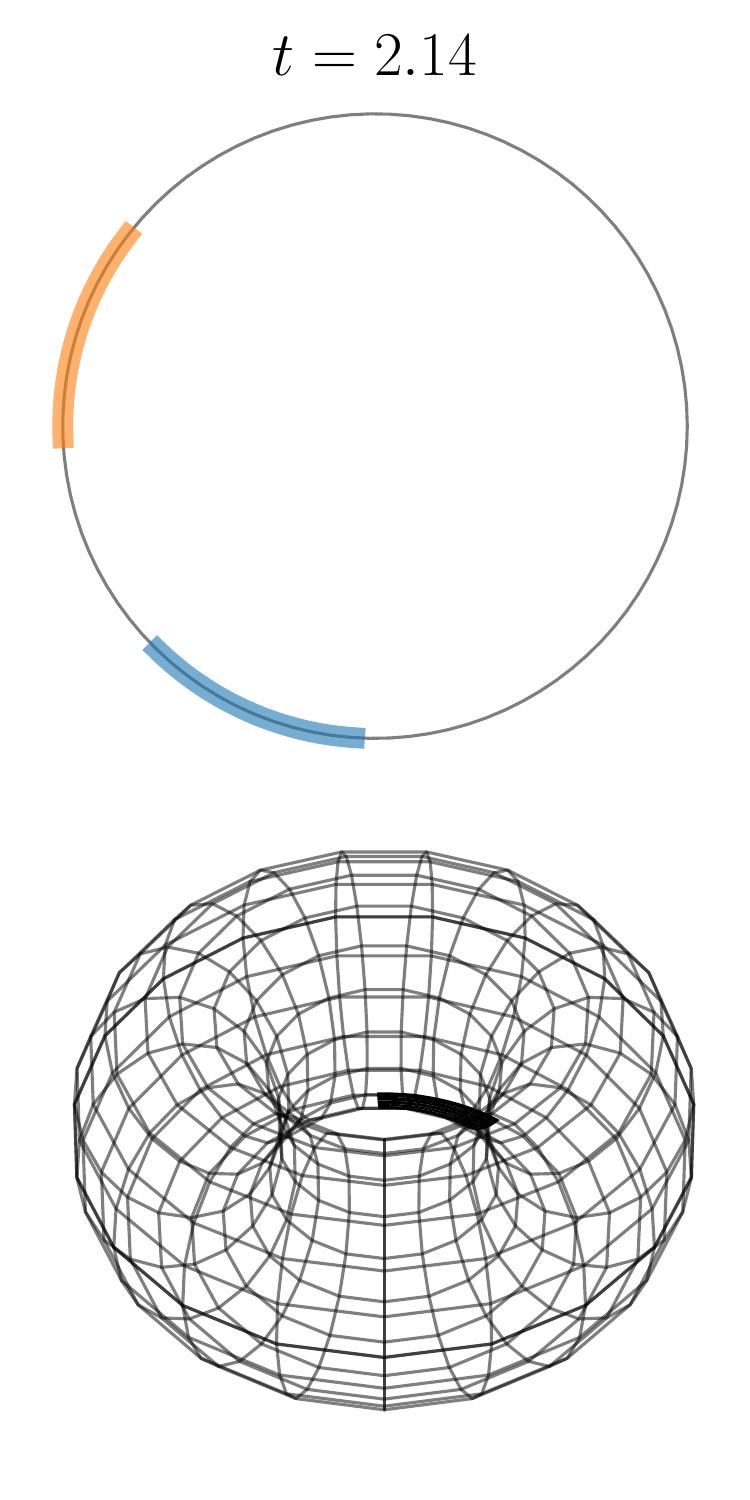}
    \includegraphics[width=0.24\columnwidth]{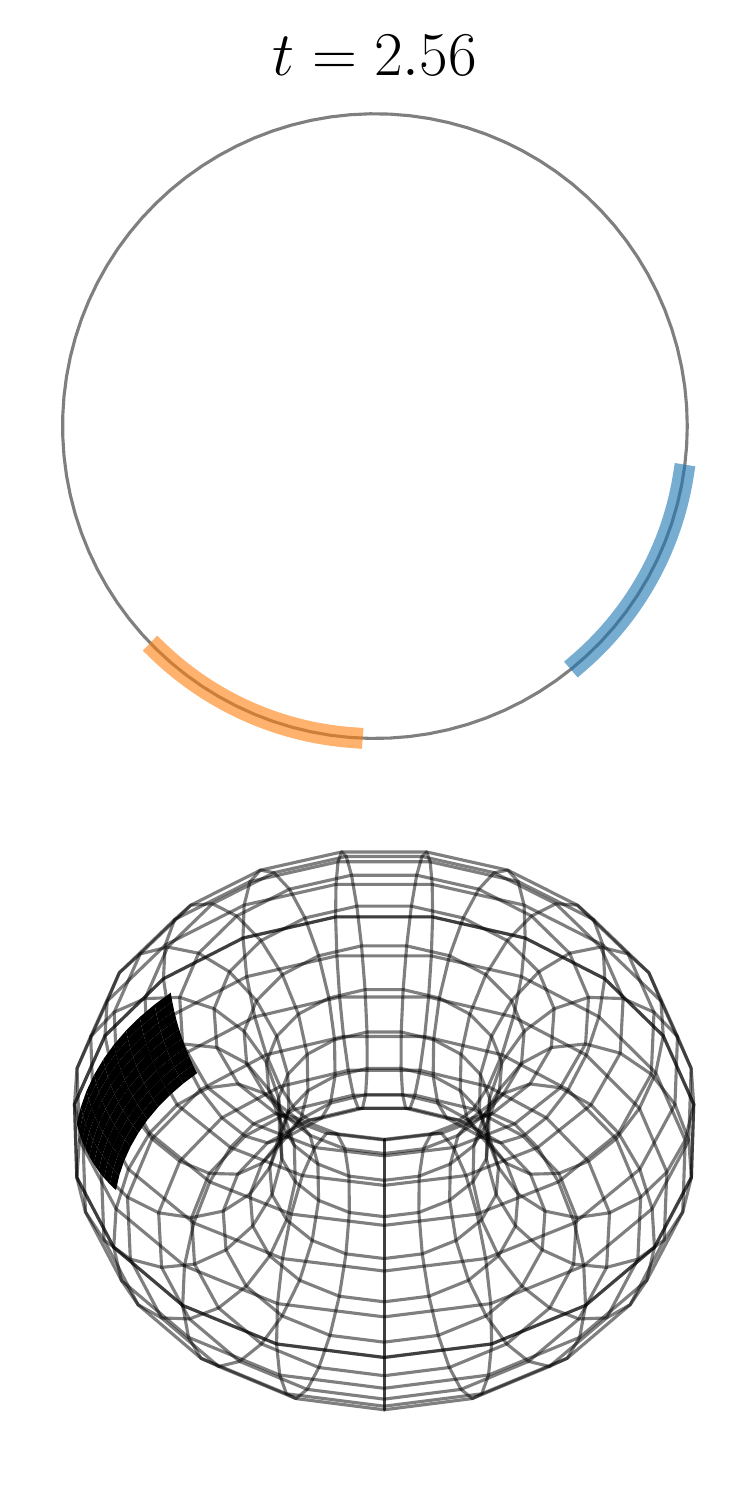}
    \includegraphics[width=0.24\columnwidth]{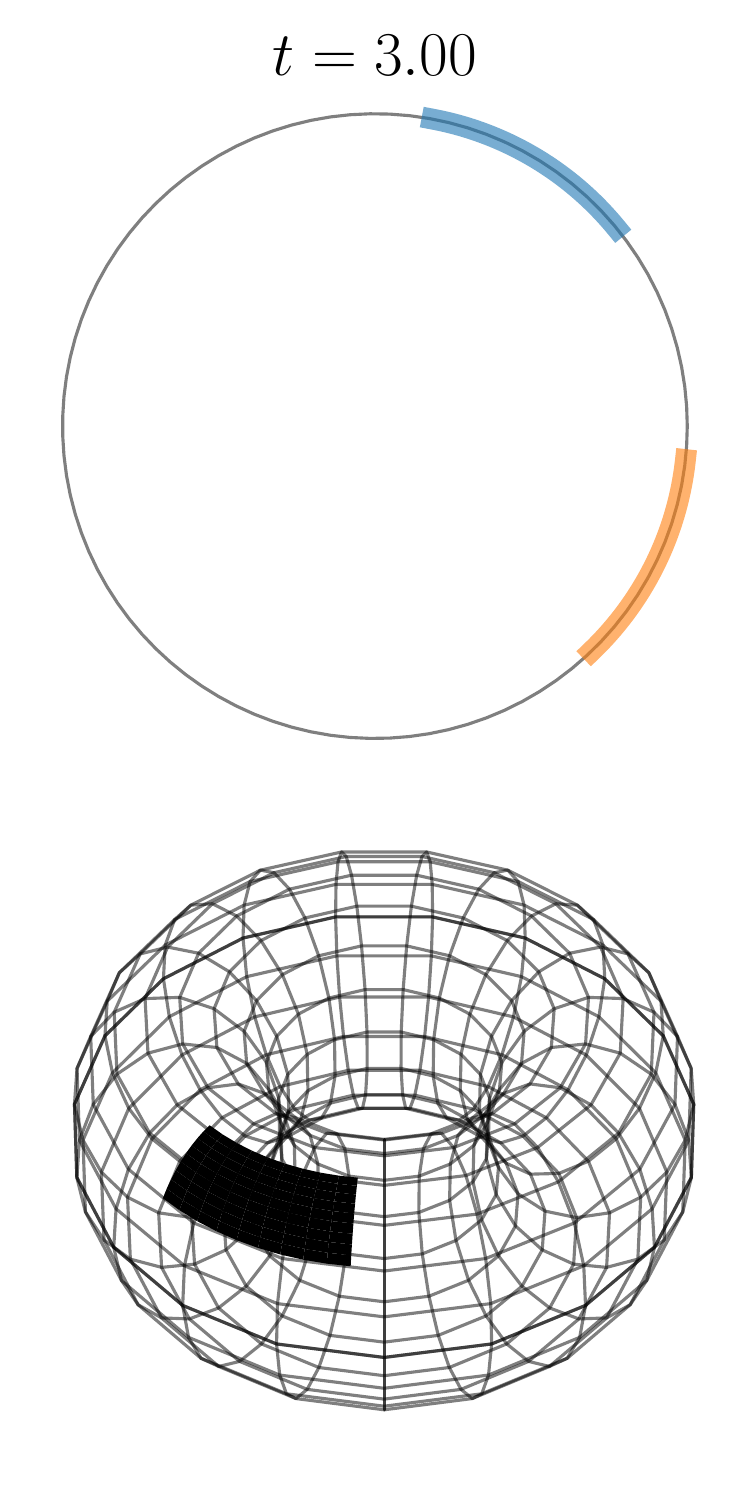}
    \caption{The reachable set $\{\mathring{x}_t\exp([\ul\Theta_t,\ol\Theta_t]_K)\}_t$ for the coupled oscillators is visualized as two arcs on a circle, as well as on a torus embedded in $\R^3$. For this system, the uncertainty in the initial condition of the blue oscillator begins larger than the orange oscillator, but the uncertainty is quickly shared between the two. For the abelian torus, the recentering can be done without any loss of information, allowing the reachable set to remain tight to the true reachable set.}
    \label{fig:torus}
\end{figure}

\subsection{Attitude control on $SO(3)$} \label{subsec:exSO3}
Consider the following control system on $SO(3)$,
\begin{align}
    \dot{R} = R\hat{u},
\end{align}
where $R\in SO(3)$ is the satellite's state, $u\in\calU = \R^3$, and the hat map $\hat{\cdot}:\R^3\to\so(3)$ is defined using the basis
\begin{align}
\begin{gathered}
    X = \begin{bsmallmatrix}
        0 & 0 & 0 \\
        0 & 0 & -1 \\
        0 & 1 & 0 
    \end{bsmallmatrix}, \
    Y = \begin{bsmallmatrix}
        0 & 0 & 1 \\
        0 & 0 & 0 \\
        -1 & 0 & 0 
    \end{bsmallmatrix}, \
    Z = \begin{bsmallmatrix}
        0 & -1 & 0 \\
        1 & 0 & 0 \\
        0 & 0 & 0 
    \end{bsmallmatrix}, \\
    \hat{u} = u_1X + u_2Y + u_3Z,
\end{gathered}
\end{align}
for $\so(3)$.
The system $\Upsilon$~\eqref{eq:liealgsys} is 
\begin{align} \label{eq:so3sys}
    \dot{\Theta}(t) &= \dexp_{\Theta}^{-1}(\hat{u}) = \dexp_{\Theta}^{-1}(u_1X + u_2Y + u_3Z)
\end{align}

Consider the left-invariant cone field $\calK$ induced by the cone $K\subset\so(3)$, where $K = \{k_1X + k_2Y + k_3Z : k_1,k_2,k_3\geq 0\}$.
The basis $\{X, Y, Z\}$ provides the identification $\so(3) \simeq \R^3$ and $K\simeq\R^3_+$. In the basis, the commutator $\dbrak{\hat{u},\hat{v}} = \widehat{u \times v}$, where $\times$ is the cross product on $\R^3$.
As such, the $\hat{\cdot}$-identified $\dexp^{-1}$~\eqref{eq:dexpinv} and BCH formula~\eqref{eq:bch} become vector valued functions $\dexp^{-1}:\R^3\times\R^3\to\R^3$ and $\operatorname{bch}:\R^3\times\R^3\to\R^3$, involving cross products. Using the natural inclusion function from \texttt{npinterval}~\cite{AH-SJ-SC:23a}, we compute inclusion functions $\textsf{DEXPINV}:\IR^3\times\IR^3\to\IR^3$ and $\textsf{BCH}_\Theta:\IR^3\to\IR^3$ by truncating the series expansions~\eqref{eq:dexpinv} and~\eqref{eq:bch} to only include fourth-order terms and below.
\textsf{DEXPINV} allows us to embed~\eqref{eq:so3sys} into an embedding system~\eqref{eq:mmembsys} using the identification from~\eqref{eq:euclididentify}.

We consider a time-varying control input $\bfu:[0,5]\to\R^3$,
\begin{align*}
    \bfu(t) = \left[\frac{5 - t}{5}, 1 - \left(\frac{t}{5}\right)^2, \sin\left(\frac{\pi t}{2}\right)\right]^T + w(t),
\end{align*}
with disturbance $w(t)\in[-0.01,0.01]^3$, and the initial set $\exp(\widehat{[-0.01,0.01]^3})\subset SO(3)$. 
We use Algorithm~\ref{alg:RKMK-Reach} with RK4 (fourth order method) and the \emph{\textsf{BCH} recentering condition} set to \textsc{True} (recentering at every time step); the results are plotted in Figure~\ref{fig:SO3fig}.
Computing the reachable set to $T=5$ seconds took $0.300 \pm 0.004$ seconds, averaged over $100$ runs.
With the \emph{\textsf{BCH} recentering condition} set to \textsc{False} (no recentering at all), monte carlo simulations show that the reachable set obtained by the algorithm fails to overapproximate the true reachable set after $0.5$ seconds.

\begin{figure}
    \centering
    \includegraphics[width=0.32\columnwidth]{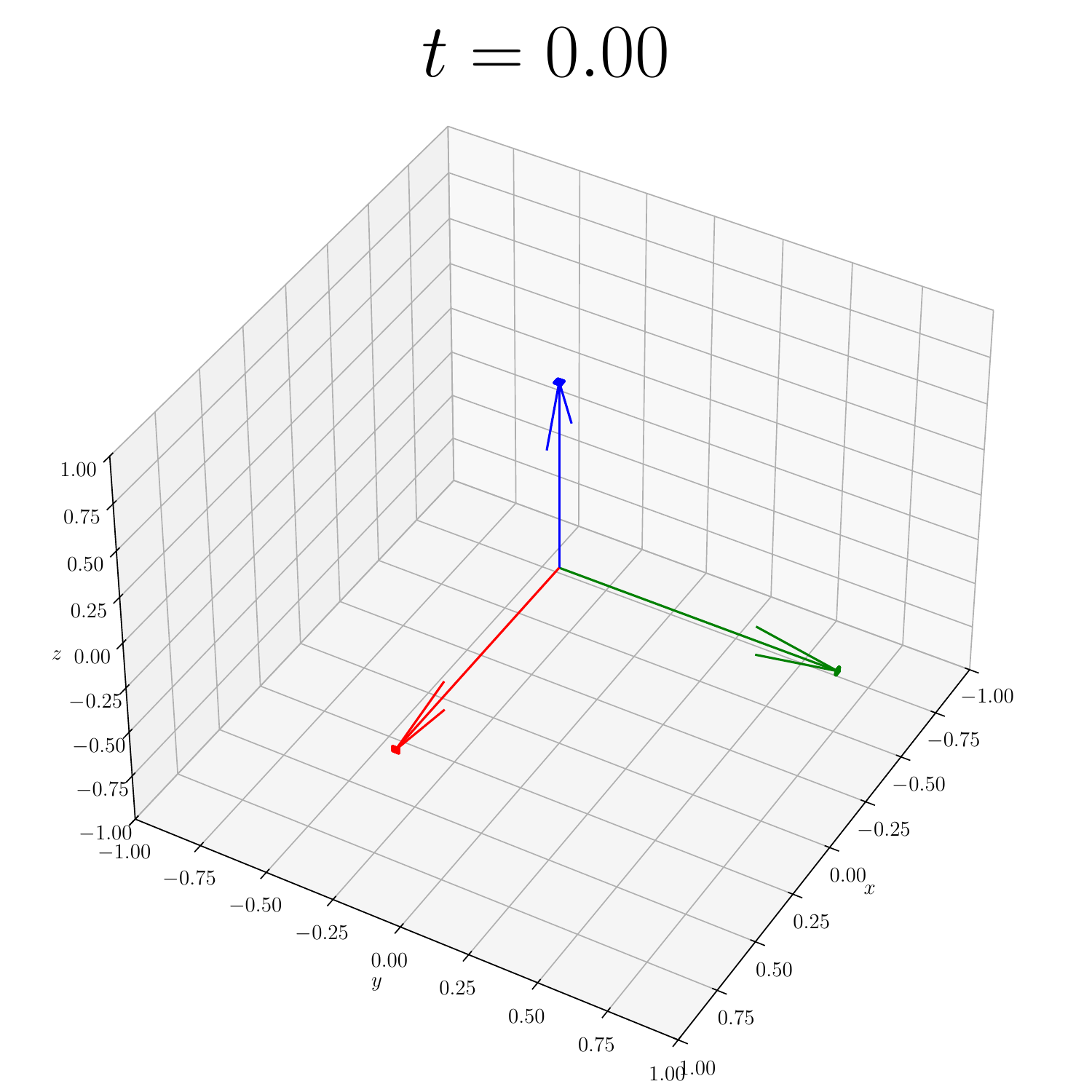}
    \includegraphics[width=0.32\columnwidth]{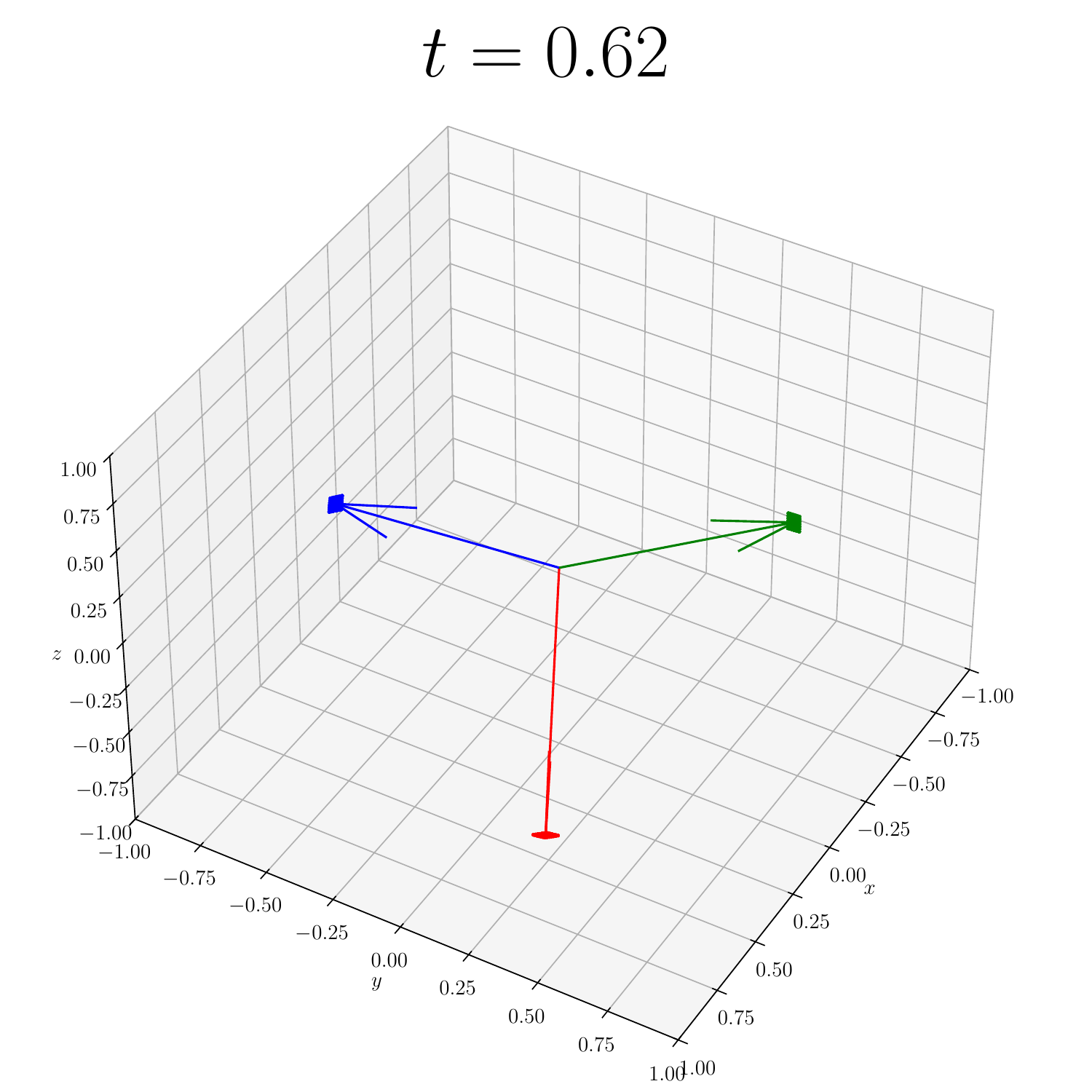}
    \includegraphics[width=0.32\columnwidth]{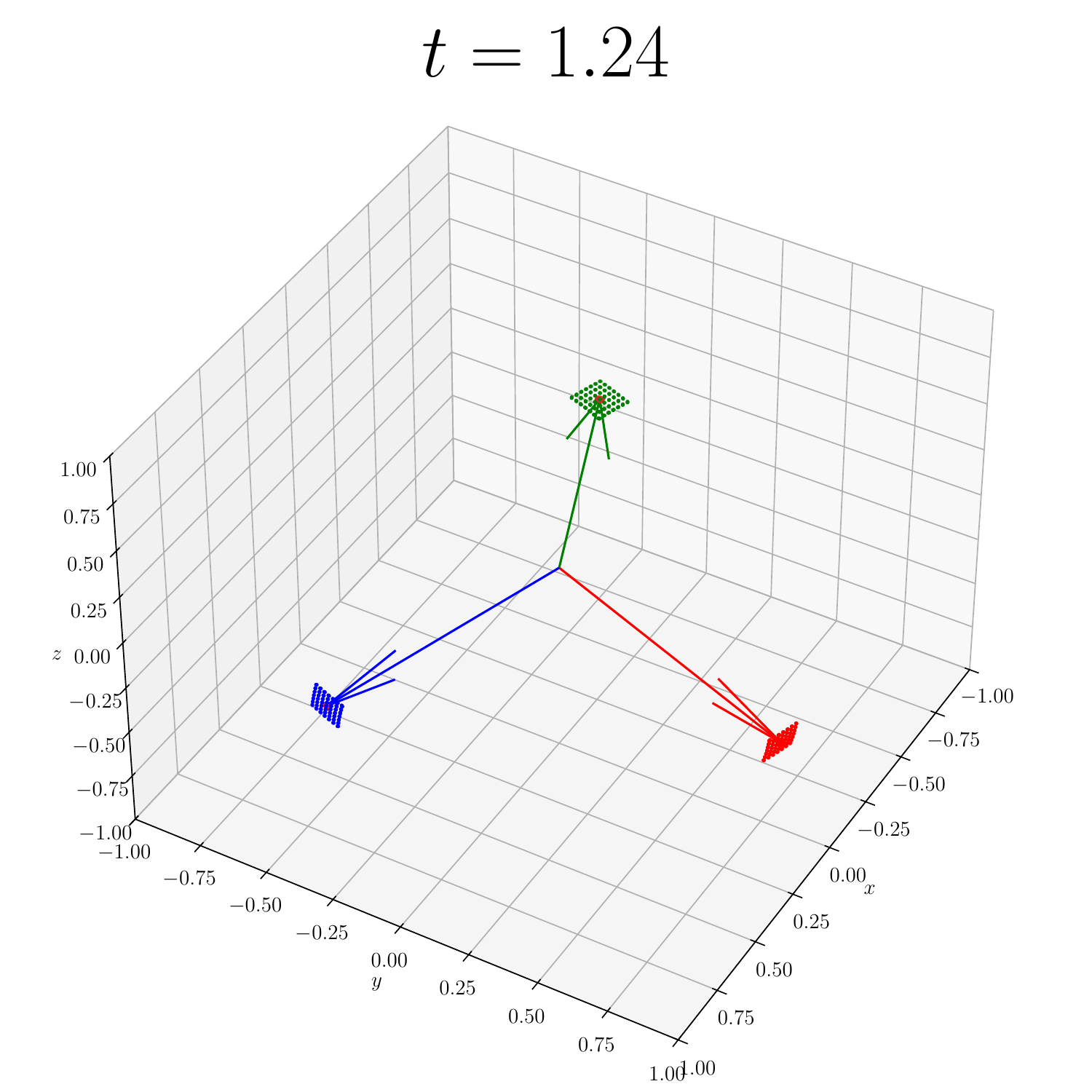}
    \includegraphics[width=0.32\columnwidth]{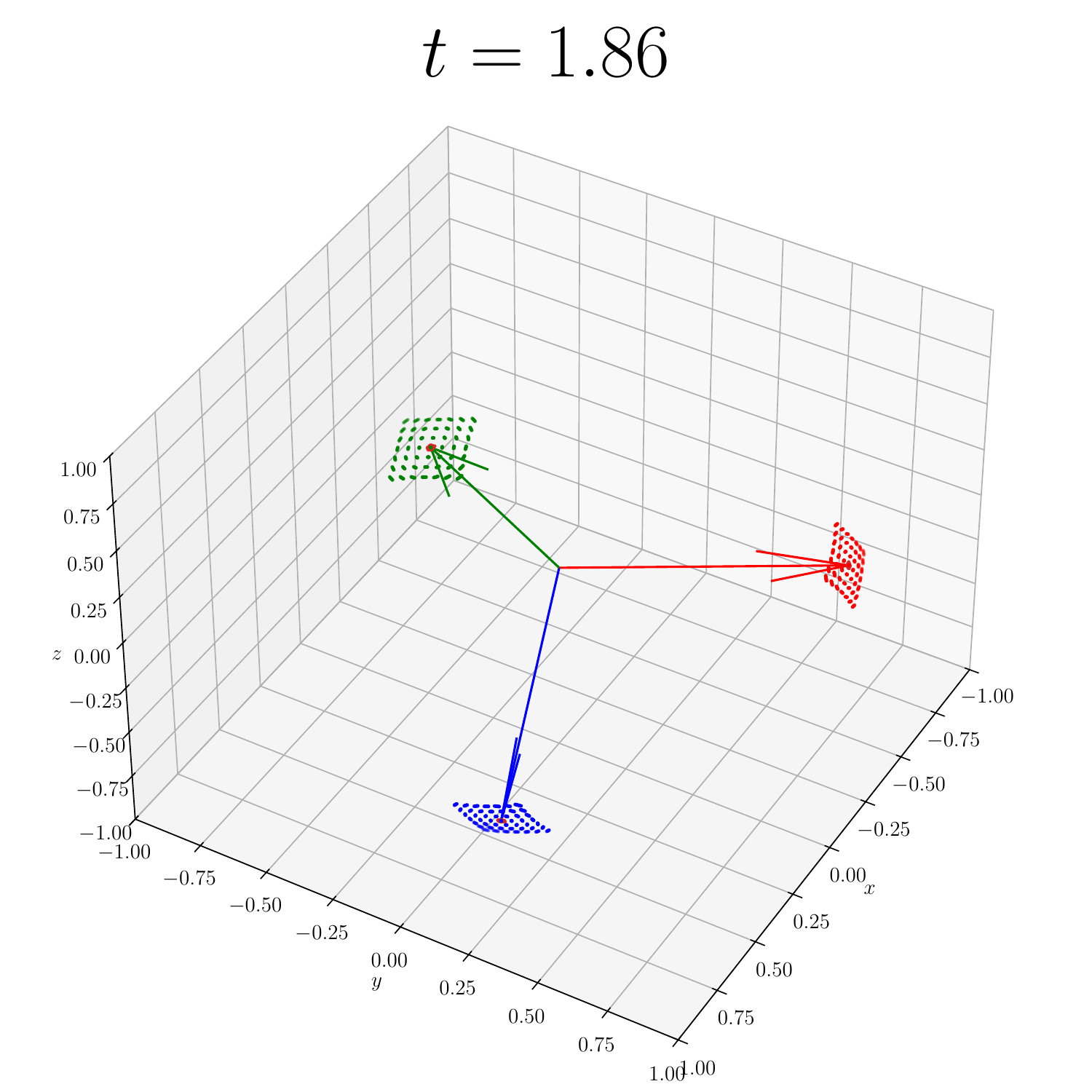}
    \includegraphics[width=0.32\columnwidth]{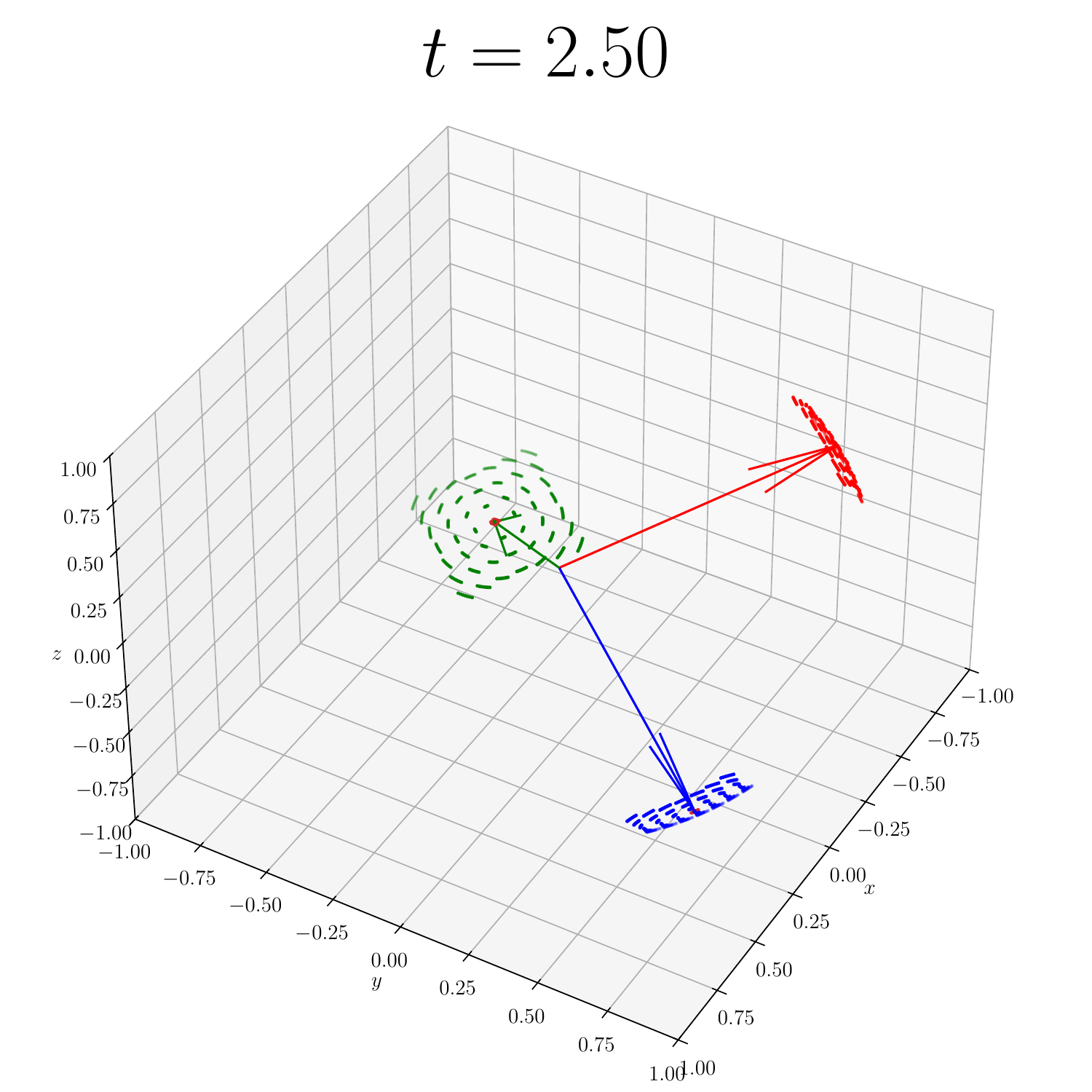}
    \includegraphics[width=0.32\columnwidth]{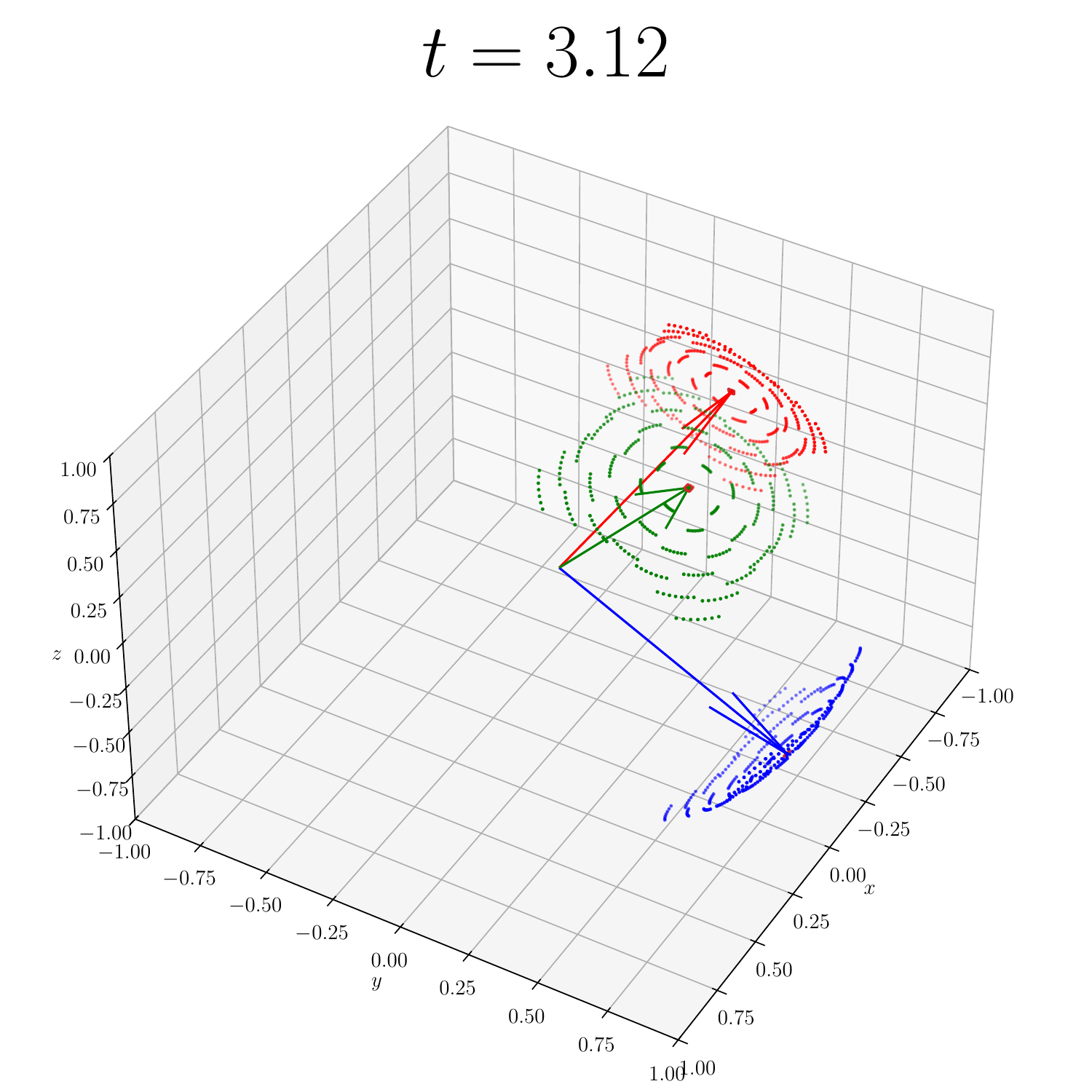}
    \includegraphics[width=0.32\columnwidth]{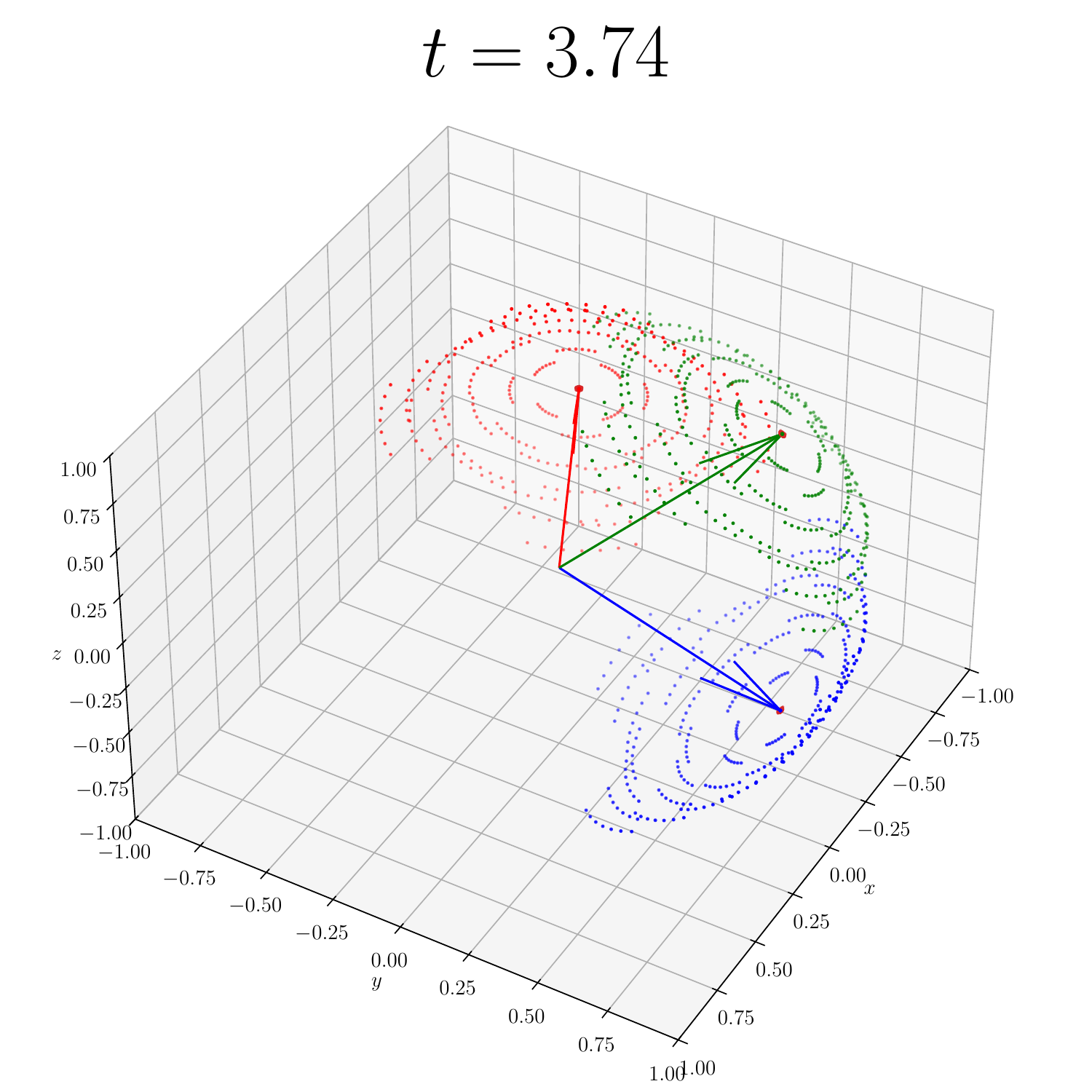}
    \includegraphics[width=0.32\columnwidth]{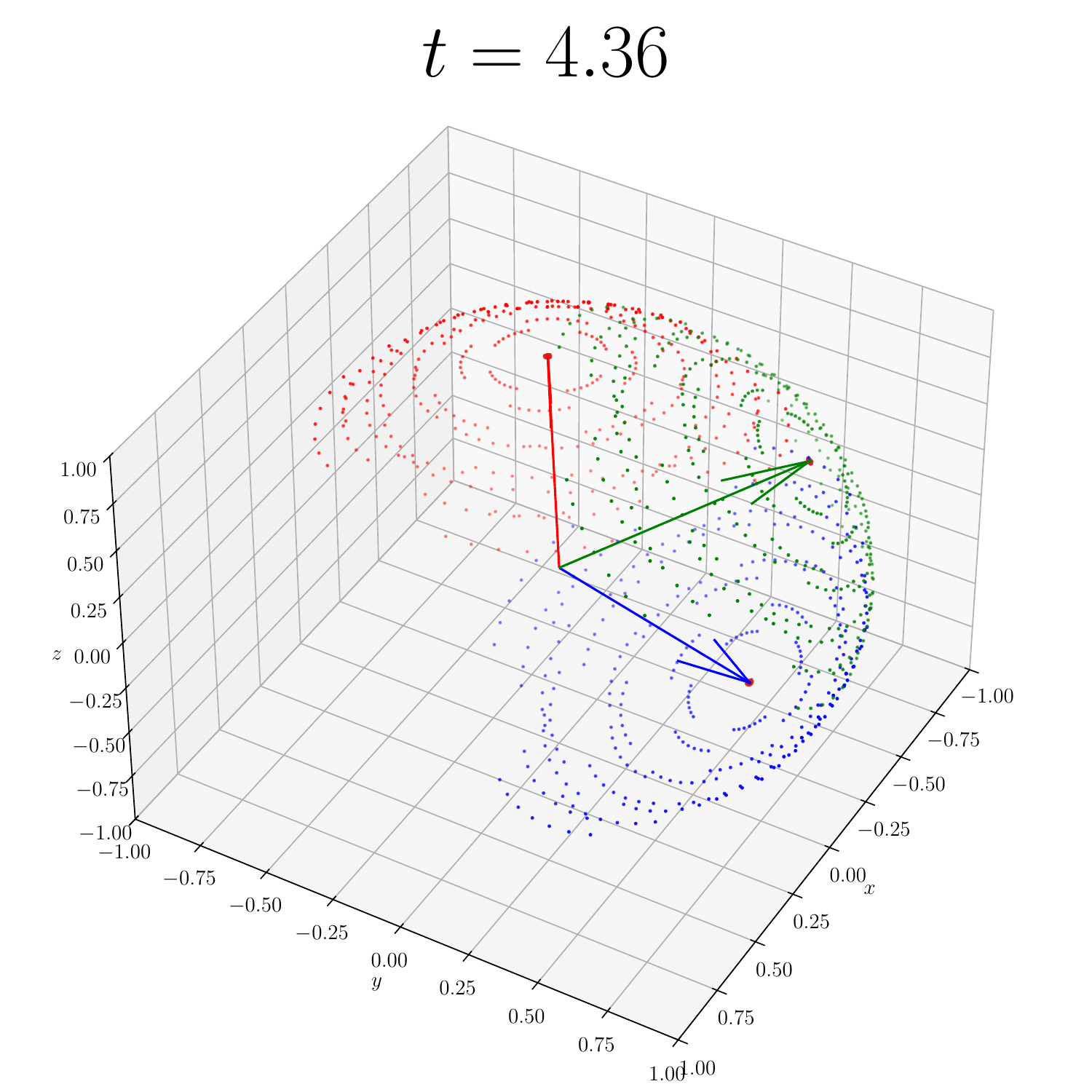}
    \includegraphics[width=0.32\columnwidth]{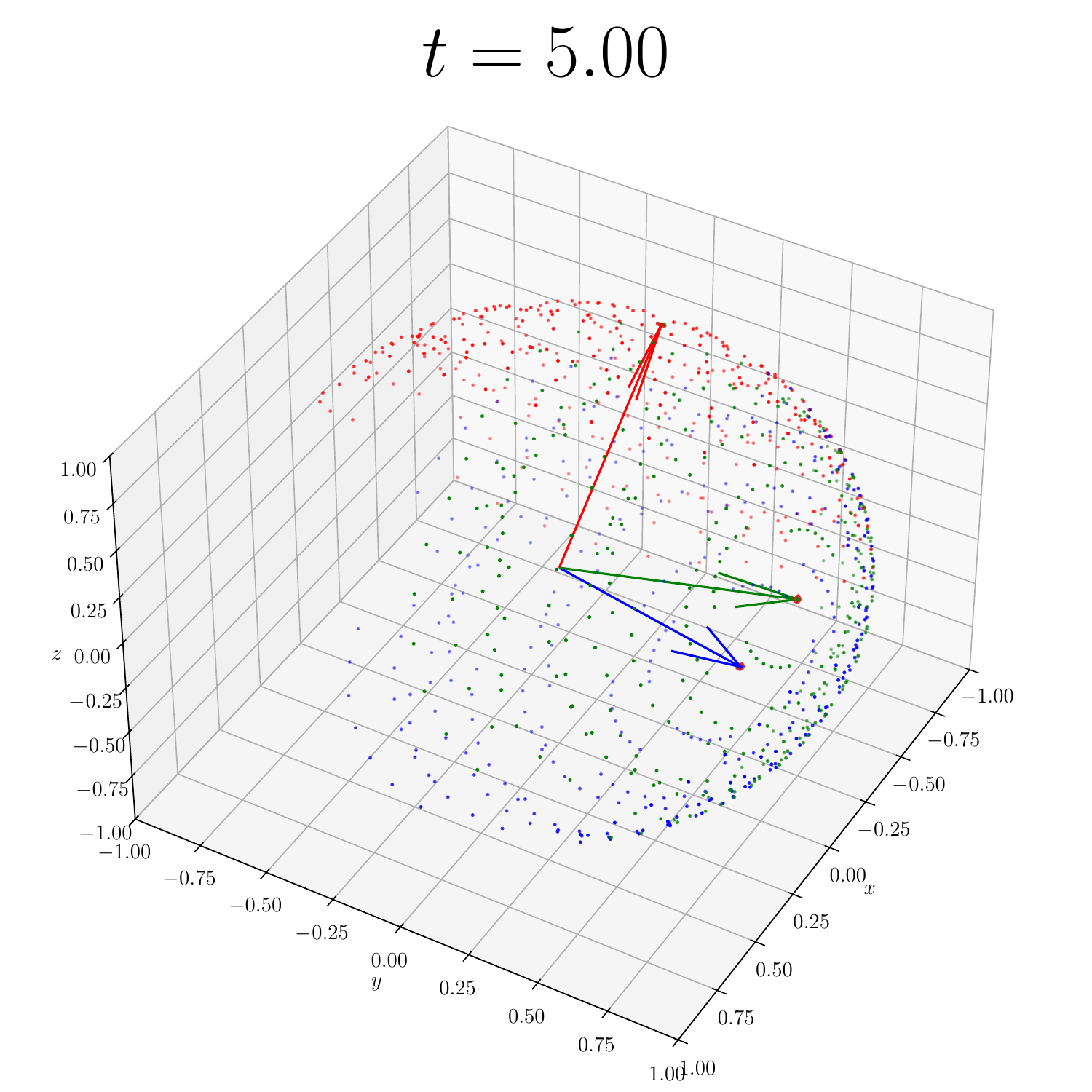}
    \caption{The overapproximated reachable set $\{\mathring{x}_t\exp([\ul\Theta_t,\ol\Theta_t]_K)\}_t$ for the orientation of a satellite evolving on $SO(3)$, for the run constantly recentering using $\textsf{BCH}$. The coordinate frame visualizes $\mathring{x}_t$, where red, green, blue represent the $x$, $y$, and $z$-axes respectively. The point clouds represent the reachable set. Each point is generated from an evenly spaced meshgrid of $7^3=343$ points $\Theta_t \in [\ul\Theta_t,\ol\Theta_t]_K$ in the Lie algebra which are exponentiated to the element $\mathring{x}_t\exp(\Theta_t)$ in the Lie group.}
    \label{fig:SO3fig}
\end{figure}

\section{DISCUSSION} \label{sec:discussion}
\subsection{Other reachability approaches}

In general, since the Lie algebra dynamics~\eqref{eq:liealgsys} evolve on a vector space, it may be possible to consider other set geometries in the tangent space beyond intervals, such as polytopes and zonotopes, allowing the use of other reachable set frameworks with different capabilities. 
Rather than simply using an inclusion function for the BCH formula, these approaches may benefit from more accurate set propogation methods, and this is an interesting direction for future research.

\subsection{Other mappings from the tangent space} \label{subsec:othermappings}

The key properties of the Lie exponential map $\exp$ used was its injectivity in the neighborhood $N_{\exp}$ around $0\in\frakg$ to a neighborhood around $e\in G$, as well as the BCH formula~\eqref{eq:bch}. In particular, the same theory applies for any other collection of mappings $\psi_p :T_pM \to M$, provided they also restrict to a diffeomorphism for some neighborhood $U_p$ of each $p$, and one can overapproximate the transition between tangent space, \emph{i.e.}, bound $h_q$ in
\[
\psi_p(v_p) = \psi_q(h_q(v_p)),
\]
with an inclusion function $\sfH_{q}$.

\subsection{Comparison to differential positivity}

Consider a manifold $\calX$ with a cone field $\calK$.
A \emph{conal curve} is a smooth curve $\gamma:I\to\calX$, such that $\gamma'(s) \in \calK(\gamma(s))$ for every $s\in I$.
A cone field $\calK$ endows the manifold $\calX$ with the \emph{conal order} $\cleq_{\calK}$, where for $x_1,x_2\in\calX$, $x_1\cleq_\calK x_2$ if and only if there exists a conal curve $\gamma:[0,1]\to\calX$ such that $\gamma(0) = x_1$ and $\gamma(1) = x_2$.
For $\ulx,\olx\in\calX$, define the \emph{conal interval} $[\ulx,\olx]_\calK:= \{x\in\calX : \ulx \cleq_\calK x \cleq_\calK \olx\}$~\cite{JL:89}.

The system defined by
\begin{align*}
    \dot{x} = f(x),
\end{align*}
where $f$ is a smooth vector field on $\calX$, is \emph{differentially positive}~\cite{FF-RS:16} with respect to $\calK$ if its linearization leaves $\calK$ invariant, \emph{i.e.}, for every $x\in\calX$ and $t>0$,
\begin{align*}
    v_x\in\calK(x) \implies (d\phi_t)_{x}(v_x)\in\calK(\phi_t(x)),
\end{align*}
where $\phi_t(x)$ is the flow from $x$.
It can be shown that differentially positive systems exhibit a similar characteristic to monotone systems through the conal order, namely that
\begin{align} \label{eq:diffposreach}
    x_0\in[\ulx_0,\olx_0]_\calK \implies \phi_t(x_0) \in [\ulx(t),\olx(t)]_\calK,
\end{align}
where $\ulx(t) = \phi_t(\ulx_0)$ and $\olx(t) = \phi_t(\olx_0)$.
However, this property is not enough to characterize reachable sets on a manifold.
In the next Example, we explore how the conal order can lose any sense of locality.

\begin{example} [Conal intervals on $\bbS$] \label{ex:conalintervals}
There are only two choices of smooth and nonvanishing cone fields on the circle $\bbS$: in each tangent space, choose either the ray pointing counterclockwise or clockwise; for smoothness to hold, all tangent spaces should be the same choice of clockwise or counterclockwise.
Between any two points $x,y\in\bbS$ on the circle, there is a curve connecting them in either the clockwise or counterclockwise direction (a conal curve), so in either cone field, $x\cleq_\calK y$ and $y\cleq_\calK x$. Any conal interval $[\ulx,\olx]_\calK$ therefore represents the entire circle since every $y\in\bbS$ satisfies $\ulx\cleq_\calK y$ and $y \cleq_\calK \olx$. So while~\eqref{eq:diffposreach} holds trivially, it is not helpful for characterizing any real reachable sets on the manifold.
\end{example}

In~\cite{FF-RS:16,CM-RS:18a,CM-RS:18b}, infinitesimal properties of the differentially positive system are considered, which subvert the use of the conal interval.
While the conal order is generally not a global partial order on $\calX$, it is a \emph{local} partial order, in the sense that for every $x\in\calX$, there exists a restriction to a neighborhood $N_x$ where the conal order partially orders $N_x$~\cite[Proposition 1.13]{JL:89}. 
This neighborhood is generally hard to find, and there has been some work in understanding the connection between the exponential map and the conal order~\cite{DK-KHN:95}.
Rather than finding this neighborhood, our formulation incurs some extra overapproximation error through the $\textsf{BCH}$ inclusion function, but allows us to simply use the exponential map from the Lie algebra of any Lie group.

\section{CONCLUSIONS}
In this paper, we developed a framework for efficient reachable set overapproximation for nonlinear systems evolving on Lie groups.
Our approach overapproximates reachable sets by evolving a tangent interval using (mixed) monotone systems theory on a locally equivalent system in the Lie algebra for small $t$, then recentering using an inclusion function for the Baker–Campbell–Hausdorff formula to remain within the injectivity neighborhood of the exponential map.
Ultimately, we developed Algorithm~\ref{alg:RKMK-Reach}, which implements this reachability scheme using Runge-Kutta-Munthe-Kaas numerical integration in the Lie algebra.
We demonstrated the approach on a consensus problem on a torus, as well as a satellite attitude control problem. 
In future work, we hope to extend this approach to apply in more general settings.

\bibliographystyle{IEEEtran}
\bibliography{diffpos}

\end{document}

%% file: macros.tex
\newcommand{\smallconc}[2]{\begin{bsmallmatrix} #1 \\ #2 \end{bsmallmatrix}}

\newcommand{\leqse}{\leq_{\mathrm{SE}}}

\newcommand{\so}{\mathfrak{so}}
\newcommand{\dbrak}[1]{\left\llbracket #1 \right\rrbracket}
\newcommand{\cleq}{\preceq}
\newcommand{\dexp}{\operatorname{dexp}}

\newcommand{\R}{\mathbb{R}}

\newcommand{\I}{\mathbb{I}}
\newcommand{\IR}{\mathbb{IR}}

\newcommand{\calK}{\mathcal{K}}

\newcommand{\calU}{\mathcal{U}}

\newcommand{\calX}{\mathcal{X}}

\newcommand{\bbI}{\mathbb{I}}

\newcommand{\bbR}{\mathbb{R}}
\newcommand{\bbS}{\mathbb{S}}

\newcommand{\bfu}{\mathbf{u}}

\newcommand{\sfE}{\mathsf{E}}
\newcommand{\sfF}{\mathsf{F}}
\newcommand{\sfG}{\mathsf{G}}
\newcommand{\sfH}{\mathsf{H}}

\newcommand{\ul}[1]{\underline{#1}}

\newcommand{\ulu}{\ul{u}}
\newcommand{\ulv}{\ul{v}}
\newcommand{\ulw}{\ul{w}}
\newcommand{\ulx}{\ul{x}}

\newcommand{\ulF}{\ul{F}}

\newcommand{\ol}[1]{\overline{#1}}

\newcommand{\olu}{\ol{u}}
\newcommand{\olv}{\ol{v}}
\newcommand{\olw}{\ol{w}}
\newcommand{\olx}{\ol{x}}

\newcommand{\olF}{\ol{F}}

\newcommand{\frakg}{\mathfrak{g}}

\newcommand{\frakx}{\mathfrak{x}}